\documentclass[11pt]{article}

\usepackage[cmex10]{amsmath}
\usepackage{amsthm, amssymb}
\usepackage{hyperref}
\usepackage[margin=1in]{geometry}

\usepackage{graphicx}

\usepackage[ruled, linesnumbered, vlined]{algorithm2e}
\SetKwInput{KwInput}{Input}
\SetKwInput{KwOutput}{Output}
\SetKwComment{tcc}{\# }{}
\SetKwComment{tcp}{\# }{}

\SetCommentSty{monoemph}

\usepackage{listings}
\usepackage[scaled=0.8]{DejaVuSansMono}
\usepackage[T1]{fontenc}
\usepackage{color}

\newtheorem{Definition}{Definition}
\newtheorem{Proposition}{Proposition}
\newtheorem{Theorem}{Theorem}
\newtheorem{Lemma}{Lemma}

\newtheorem{Remark}{Remark}

\usepackage{multirow}

\title{New data structure for univariate polynomial approximation and applications to
root isolation, numerical multipoint evaluation, and other problems}

\author{Guillaume Moroz}
\date{%
  Université de Lorraine, CNRS, Inria, LORIA, F-54000 Nancy, France\\
  guillaume.moroz@inria.fr\\[2ex]
  November 17, 2021 }

\begin{document}
\maketitle

\begin{abstract}
  We present a new data structure to approximate accurately and
  efficiently a polynomial $f$ of degree $d$ given as a list of
  coefficients $f_i$.  Its properties allow us to improve the state-of-the-art
  bounds on the bit complexity for the problems of root isolation and
  approximate multipoint evaluation. This data structure also leads to a
  new geometric criterion to detect ill-conditioned polynomials,
  implying notably that the standard condition number of the zeros of a
  polynomial is at least exponential in the number of roots of modulus
  less than $1/2$ or greater than $2$.

  Given a polynomial $f$ of degree $d$ with $\|f\|_1 = \sum |f_i| \leq 2^\tau$ for
  $\tau \geq 1$, isolating all its complex roots or evaluating it at $d$
  points can be done with a quasi-linear number of arithmetic
  operations. However, considering the bit complexity, the
  state-of-the-art algorithms require at least $d^{3/2}$ bit operations
  even for well-conditioned polynomials and when the accuracy required
  is low.  Given a positive integer $m$, we can compute our new data
  structure and evaluate $f$ at $d$ points in the unit disk with an
  absolute error less than $2^{-m}$ in $\widetilde O(d(\tau+m))$ bit
  operations, where $\widetilde O(\cdot)$ means that we omit logarithmic
  factors.  We also show that if $\kappa$ is the absolute condition
  number of the zeros of $f$, then we can isolate all the roots of $f$
  in $\widetilde O(d(\tau + \log \kappa))$ bit operations. Moreover, our
  algorithms are simple to implement. For approximating the complex
  roots of a polynomial, we implemented a small prototype in
  \verb|Python/NumPy| that is an order of magnitude faster than the
  state-of-the-art solver \verb/MPSolve/ for high degree polynomials
  with random coefficients.
\end{abstract}


\section{Introduction}
% TODO: say that precision required is always greater than degree for
% root-finding or multipoint evaluation.
% 
% TODO: observe that in the disk of radius one half, things are easier, so
% reduce to this case often. (Graeffe iteration, Ritzmann, van der Hoeven,
% ...).
% 
% TODO: example of Graeffe iteration : either coefficient growth, or
% alternative with quadratic number of operations.
% 
% TODO: as we will see Graeffe iteration without normalization produces ill-conditioned
% polynomials, so not necessarily the best approach from a numerical point
% of view.
% 
% TODO: start with the bottleneck problem of multipoint evaluation with error of constant size
% in enhance form.
% 
% 
% TODO: say that there are many workaround in the literature
% 
% TODO: say that we present a rootfinding algorithm but that our data
% structure can be used and improve many existing approaches.
% 
% TODO: say that it has application to multivariate case, homotopy
% continuation where evaluation is a bottleneck
% 
% TODO: say it opens the door to a new point of view on polynomial system
% solving. maybe

One of the fundamental problem in computer algebra is the evaluation of
polynomials. Since 1972, it is known that evaluating a univariate polynomial of
degree $d$ on $d$ points can be done in a quasi-linear number of
arithmetic operations \cite{Fstoc72}. Unfortunately, this bound
doesn't hold if we consider the bit complexity, where the arithmetic
operations performed with a precision of $m$ bits costs $\widetilde O(m)$ bit
operations. If we want to evaluate approximatively a polynomial on $d$ points up to a constant absolute error, a direct application of
Fiduccia algorithm leads to a bit-complexity bound in $\widetilde O(d^2)$ bit
operations, and a more sophisticated algorithm provides a bound in
$\widetilde O(d^{3/2})$ \cite{Hrr08}.  For almost 50 years, the
following problem has remained open.

% the precision
% required for the arithmetic
% operations performed by the algorithm of Fiduccia may be linear in the
% degree. Thus, this algorithm is quadratic in the number of
% bit operations. This problem becomes a bottleneck in other problems of
% computer algebra. For example, to isolate the
% roots of a polynomial, state-of-the-art algorithms check that the disk are isolating by evaluating a
% polynomial of degree linear in $d$ on the approximated roots computed
% (\cite[\S 2.2.3]{}, Melhorn, \cite[], Yap)
% TODO

% handling polynomials numerically with finite precision difficult.

\begin{quote}
  Given a polynomial $f$ of degree $d$ with coefficients of constant
  size, and $d$ complex points $x_k$ in the unit disk, is it possible to
  compute all the $f(x_k)$ up to a constant absolute error with a number of
  bit operations quasi-linear in $d$ ?
\end{quote}

%to find an algorithm quasi-linear in $d$ that returns the
%numerical evaluation of a polynomial $f$ of degree $d$ on $d$ points
%$x_k$ in the unit disk. 
Nevertheless, the evaluation of polynomials on
multiple points is used in many areas of computer science, such as 
polynomial system solving with the Newton method
\cite{Sbams81,HSSim01,Dbook06,Bbook13, BASmc16}, homotopy continuation
\cite{CSacm99,CKMWjc08,BPfocm11,Lfocm17,Bbook13} or subdivision
algorithms \cite{Nbook90, Kbook96, MKCbook09}, visualisation of
algebraic surfaces through raytracing
or mesh computation \cite{Wbook13}, among others.
Speeding up the numerical evaluation of polynomials may lead to
an effortless practical improvement for many existing algorithms.

We introduce in Sections~\ref{sec:introdatastructure}
and~\ref{sec:datastructure} a new data structure, that allows us to finally
solve this problem. It approximates the input polynomial by a
piecewise polynomial on a carefully chosen domain
that depends solely on the degree of the input polynomial, and the
required precision. The fact that the domain is fixed and not adaptive
makes it easy to implement.

Moreover, we show that our new data structure improves not only the
bound for the numeric multipoint evaluation problem
(Sections~\ref{sec:introevaluation} and~\ref{sec:evaluation}), but also the bound
for the root isolation problem (Sections~\ref{sec:introrootfinding}
and~\ref{sec:rootfinding}), and the lower bound on the condition
number of polynomials (Sections~\ref{sec:introcondition}
and~\ref{sec:condition}). We expect that our approach will lead further
improvements for other related problems, notably multivariate polynomial
evaluation and polynomial system solving. As a proof of concept, we also
implemented the root isolation algorithm presented in this article. Even
though our implementation is less than $100$ lines of code written in
\verb/Python/, it is an order of magnitude faster than the
state-of-the-art optimized implementation \verb/MPSolve/ for high degree
polynomials with random coefficients (Section~\ref{sec:introexperiment}
and source code in Appendix~\ref{apx:source}).

% -----
% 
% Evaluating and finding the roots of polynomials are two very old
% problems that have been extensively studied. Yet
% 
% TODO: mention complexity quasi-linear in $d$ for the number of
% arithmetic operations, but quadratic in $d$ for bit operations, either
% because of precision in omega d or because it requires $d^2$
% evaluations.
% 
% TODO: say that for root finding, this is the first improvement on the
% bit complexity bound since 2002.
% 
% TODO: say that it give a better comprehension of the link between the
% roots and the condition number of a polynomial
% 
% We will show how to
% improve both problems using a new data structure to approximate a
% complex polynomial on the unit disk. As we will see, extending the data
% structure to handle the evaluation in the full complex plane can be done
% easily using the transform $g(x) = x^d f(1/x)$ where $d$ is the degree of
% $f$.
% 
% Blabla model of computation (finite field, exact integers and rationals,
% floating points), and define bit-operations.
% 
% Define absolute error.
% 
% TODO: mention great experiment timings

\subsection{The data structure}
\label{sec:introdatastructure}

%TODO: mention uniform annuli leading to power 3/2

Given a polynomial $f$ of degree $d$ and an integer $m>1$, we will
introduce in Definition~\ref{def:approximation} the so-called an
\emph{$m$-hyperbolic approximation} of $f$, which can be seen as a
piecewise approximation of $f$ by polynomials of degree $m-1$. The key
that will allow us to improve the state-of-the-art complexity bounds of
several classical problems related to univariate complex polynomials is
the hyperbolic layout used to compute this piecewise approximation. We
first define this layout so-called \emph{hyperbolic covering},
illustrated in Figure~\ref{fig:hyperbolic}.
%Roughly, it is a union of
%disks covering 
%Notably, each ring between the circles centered at $0$ and of
%radius
Roughly, a hyperbolic covering is a set of disks of radius exponentially
smaller near the unit circle, and such that their union contains the
unit disk.
% for an integer $N$, we can partition the unit disk in $N$ concentric rings $R_n$ of width
% exponentially smaller near the unit circle. A $N$-hyperbolic covering is
% a set of disks such that each ring $R_n$ is contained in a union of
% disks of width proportional to the width of $R_n$. More precisely, for
% $n < N$ letting
% $r_n = 1 - \frac 1 {2^n}$ and $r_{n+1} = 1 - \frac 1 {2^{n+1}}$,  we can see that the sequence of radii $r_n$ 

% Notably, letting $r_n = 1-\frac 1 {2^n}$ and $r_{n+1} = 1 - \frac 1
% {2^n+1}$, each ring $D(0,r_{n+1})\setminus D(0,r_n)$ is contained in a
% union of disks radius proportional to the width of the ring
% $r_{n+1}-r_n$.

\begin{Definition}
  \label{def:covering}
  Given a positive integer $N$, an \emph{$N$-hyperbolic covering} of the unit
  disk is the set of disks of centers $\gamma_ne^{i2\pi \frac k {K_n}}$
  and radii $\rho_n$ for $0 \leq n < N$ and $0 \leq k < K_n$, where
  $\gamma_n, \rho_n$ and $K_n$ are defined by:
  \begin{align*}
    r_n & = \begin{cases} 1 - \frac 1 {2^n} & \text{if $0 \leq n < N$}\\
                      1                         & \text{if $n=N$}
            \end{cases}\\
    \gamma_n &= \frac 1 2 (r_n+r_{n+1})\\
    \rho_n &= \frac 3 4 (r_{n+1}-r_n)\\
    K_n &= \begin{cases} 4 & \text{if $n = 0$}\\
                 \lceil \frac{3\pi}{\sqrt 5} \frac{r_{n+1}}{\rho_n} \rceil & \text{otherwise}
            \end{cases}
  \end{align*}
\end{Definition}
\begin{Remark}
  \label{rem:explicit}
  We can also write explicitly the corresponding sequences $(\gamma_n)_{n=0}^{N-1}$
  and $(r_n)_{n=0}^{N-1}$:
  \begin{align*}
    \gamma_n &= \begin{cases} 1 - \frac 3 4 \frac 1 {2^n} & \text{if $0 \leq n \leq N-2$}\\
                         1 - \frac 1 2 \frac 1 {2^n} & \text{if $n=N-1$}
           \end{cases}\\
    \rho_n &= \begin{cases} \frac 3 8 \frac 1 {2^n} & \text{if $0 \leq n \leq N-2$}\\
                         \frac 3 4 \frac 1 {2^n} & \text{if $n=N-1$}\\
           \end{cases}
  \end{align*}
  
  We also have explicitly $2^{n+1} \leq K_n \leq 2^{n+4}$, using
  $\frac 1 2 \leq r_{n+1} \leq 1$ and $\frac 3 8 \frac 1 {2^n} \leq \rho_n
  \leq \frac 3 4 \frac 1 {2^n}$.
\end{Remark}

We will see in Lemma~\ref{lem:disks} that for all integers $N \geq 1$,
the union of the disks of a $N$-hyperbolic covering contains the
unit disk.% As an illustration, Figure~\ref{fig:hyperbolic} shows a
%$5$-hyperbolic covering.
%The name hyperbolic comes from the fact that
%the disks are roughly uniformly spread according to a hyperbolic metric.
We can now define the $m$-hyperbolic approximation of a
polynomial $f$, that can be seen as piecewise approximation of $f$ by
polynomials of degree lower than $m$.

%TODO: redo figure with new $K_n$.

\begin{figure}
  \centering
  \includegraphics[width=\textwidth]{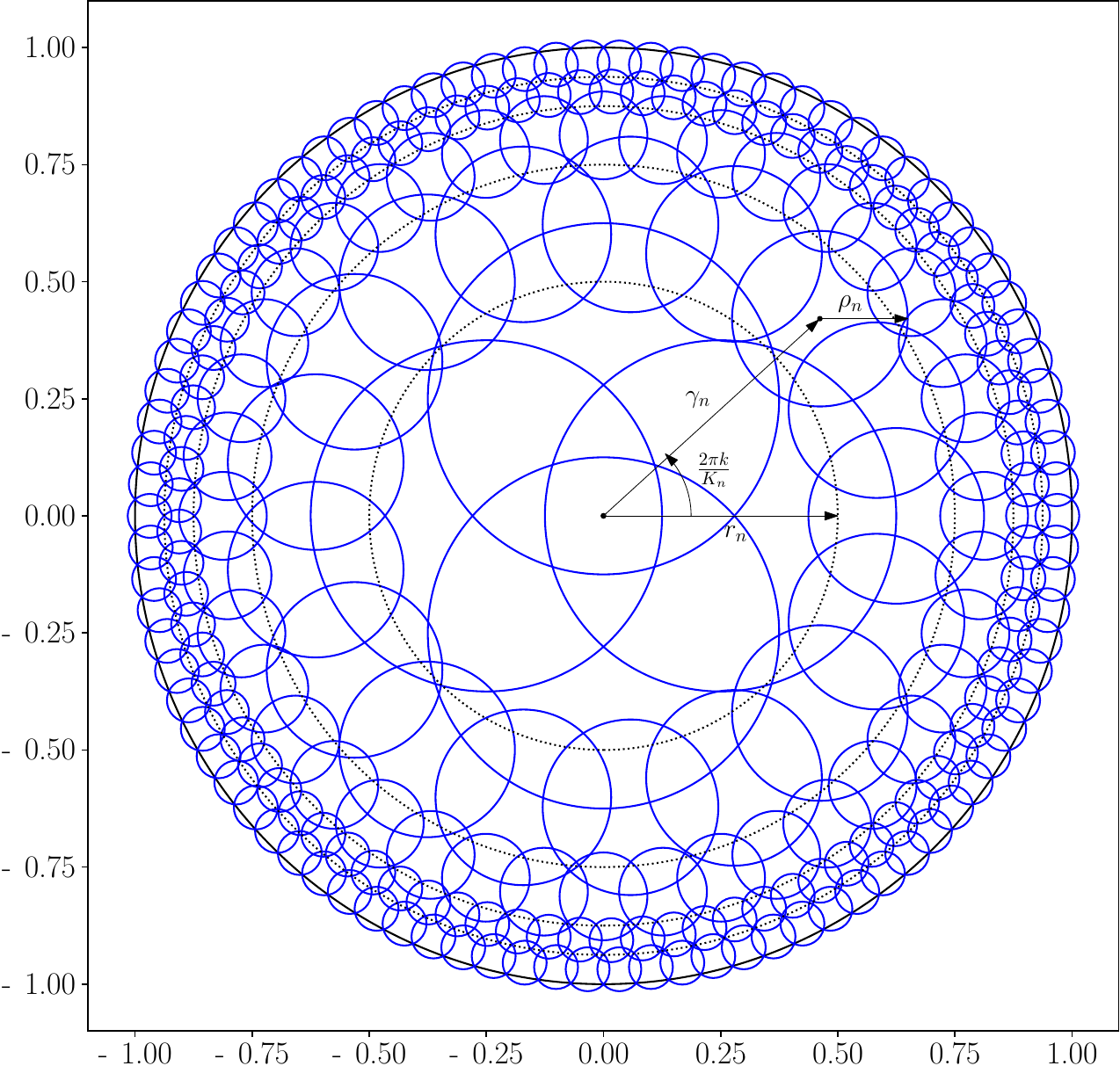}
  \caption{$5$-hyperbolic covering}
  \label{fig:hyperbolic}
\end{figure}

\begin{Definition}
  \label{def:approximation}
  Given a polynomial $f$ of degree $d$ with $\|f\|_1 \leq 2^\tau$, and an integer $m>1$, an
  \emph{$m$-hyperbolic approximation} of $f$ is a finite set of pairs $(g, a)$,
  where $g$ is a polynomial of degree $\widetilde m = \min(m-1,d)$, with coefficients of
  bit size $O(\tau+m)$, and $a$ is an affine transform, such that:
  %bit size $O(m)$, and $a(X) = (\gamma + \rho X)e^{i2\pi\frac k/M}$ is an affine transform, such that:
  %for all
  %$x$ in the unit disk $D(0,1)$:
  
  %TODO: introduce $\widetilde m$ for correct definition.
  \begin{itemize}
    \item the set of disks $a(D(0,1))$ is the $N$-hyperbolic covering
      with $N = \lceil\log_2\left(\frac {3ed} {\widetilde m}\right)\rceil$ %$D(0,1) \subset \cup_k a_k(D(0,1))$
    %\item there exists $k$ such that $x \in a_k\left(D(0,1)\right)$
    %\item $|f(x) - g_k(a_k^{-1}(x))| < 4\|f\|_1 2^{-m}$
    \item $\|f \circ a - g\|_1 \leq 3\|f\|_1 2^{-m}$
  \end{itemize} 
  %and for all $k$, the polynomials $g_k$ and $f(a_k(X)) = \sum_{\ell=0}^d
  %c_{k,\ell} X^\ell$ satisfy:
  %\begin{itemize}
  %  \item $\|f(a_k(X)) - g_k(X))\|_1 \leq 3\|f\|_1 2^{-m}$
  %  %\item $|c_{k,\ell}|\leq \|f\|_1/2^\ell$ for all $\ell \geq m$
  %\end{itemize}
  %TODO: move last item to a separate lemma.
\end{Definition}

Since we constrain the polynomials $g(X)$ approximating $f(a(X))$ to
have a degree $m-1$, that can
be significantly smaller than $d$, it is
not obvious that an approximation satisfying the conditions defined
above always exists. The following theorem proves that it always exists,
and furthermore that it can be computed in quasi-linear time with
respect to the degree of $f$.

\begin{Theorem}
  \label{thm:datastructure}
  Let $f$ be a polynomial of degree $d$ with $\|f\|_1\leq 2^\tau$, and $m>1$ be an integer. 
  Algorithm~\ref{alg:datastructure} computes an $m$-hyperbolic
  approximation of $f$, denoted by $H_{d,m}(f)$, in $\widetilde O(d(m+\tau))$ bit operations.
  %TODO: check precise bound in term of $m$
  %We can compute the data structure in quasi linear time in $d$.
\end{Theorem}

Remark also that the maximal precision required for the arithmetic
operations in Algorithm~\ref{alg:datastructure} is in $O(\tau + m+\log
d)$, which makes it suitable for an implementation with machine precision arithmetic.

\paragraph{Main idea of the proof of Theorem~\ref{thm:datastructure}.}
Denoting the center and the radius of a disk in a
$m$-hyperbolic approximation by $\gamma_{n,k}=\gamma_{n}e^{i2\pi k
/K_n}$ and
$\rho_{n}$ , we need to prove that it is possible to
compute a polynomial of degree $m-1$ that satisfies the bounds of
Theorem~\ref{thm:datastructure}. This comes from the fact that using the
formula of Definition~\ref{def:covering} and Remark~\ref{rem:explicit},
we can check that the coefficients of degree $\ell$ of the polynomial
$f(\gamma_{n,k} + \rho_n X)$ have a modulus less than $\|f\|_1/2^\ell$
for all $\ell \geq m$. Then it remains to prove that we can approximate the first $m$
coefficients of all the $f(\gamma_{n,k}+\rho_n X)$ in $\widetilde O(d(m+\tau))$ bit
operations. For that, remark that $0 \leq n < N = O(\log d)$. Thus, it
is sufficient to prove that for a fixed $n$, we can compute 
$f(\gamma_{n,k} + \rho_n X) \mod X^m$ for all $k$ in $\widetilde
O(d(m+\tau))$ bit operations. This can be done by using a combination of fast
numerical composition of series (Proposition~\ref{pro:composition}), and
numerical fast Fourier transform (Proposition~\ref{pro:fft}). More
details can be found in Section~\ref{sec:datastructure}.

% We will denote by $H_{d,m}(f)$ the $m$-hyperbolic approximation returned
% by Algorithm~\ref{alg:datastructure}.

%TODO: insert figure with N=5, also here or in the rootfinding section, the data structure
%extended on the full complex plane.

% TODO: sketch the main idea that on the disk from a hyperbolic covering,
% the polynomial $f(a_k(X))$ has exponentially decreasing coefficients.
% Also that the covering being symmetric by rotation, we can use fast
% Fourier transform to compute the truncated Taylor shifts.
% 
% TODO: mention that the algo for best asymptotic bound can be slightly
% changed for best practical efficiency. In particular fast multipoint
% evaluation is not necessary when the desired accuracy is logarithmic in
% the degree.
% 
% Refer to algo in detailed section.
% 
% Corresponding theorem.
% 
% Examples for the real case.

% \paragraph{Weyl norm}
% 
% TODO: Say that other norms are possible, beyond the scope of this
% article.
% 
% Small note to latest section

% TODO: sketch of our techniques and why it works : hyperbolic annuli
% instead of uniform ones as in van der hoeven.

\subsection{Applications}

%TODO: main idea to reduce problem to several problems of degree the
%precision required.

Based on our new data structure, we describe three independent
results that improve state-of-the-art solutions to long-standing
problems. First we improve the complexity for evaluating numerically polynomials on
multiple points. Then we improve the complexity of finding the roots of
well-conditioned polynomials. Finally, we present a new lower bound on
the condition number of the zeros of a polynomial, based on simple
geometric properties of the distribution of its roots.

\subsubsection{Numerical multipoint evaluation}
\label{sec:introevaluation}

%TODO: add algo

In the literature \cite{Fstoc72, AHbook74}, \cite[Chapter~10]{GGbook13},
a fast multipoint evaluation algorithm has been
designed to evaluate $d$ points of a degree $d$ polynomial with a number
of arithmetic operations quasi-linear in $d$. However, in the
case of numerical evaluation with precision $m$ after the binary point,
this algorithm uses a number of bit operations quadratic in $d$, even
for a constant $m$. This is due notably to the fact that this algorithm
introduces intermediate polynomials with coefficients that may have a
bit-size linear in $d$. A careful analysis of the bit-complexity of this
algorithm for numerical evaluation (\cite[Lemma~11]{Hrr08}) shows that
this algorithm uses $\widetilde O(d(m+d))$ bit operations.

Specific sets of
points were found where the problem of numeric multipoint evaluation can be solved in a quasi-linear
time in $d$.
The most famous one is the set of roots of unity. Computing a numerical
approximation of the $f(x_k)$ for $x_k = e^{i2\pi\frac k {d+1}}$ can be
done in quasi-linear time \cite{Sca82} in $d$.  Another family of points
was used by Ritzmann for the problem of fast numeric composition of
series \cite[Proposition 3.4]{Rtcs86}. He showed that if the modulus of
the $x_k$ is lower than $\frac 1 {3d}$, then all the $f(x_k)$ can be
approximated numerically in a quasi-linear number of bit operations.
%That was a key part of a fast numeric algorithm for the composition of
%series.

%In a more recent work (\cite[\S 3.2]{}, Joris), van der
%Hoeven presented several
%other sequence of points that can be evaluated in quasi-linear time.

In a more recent work \cite[\S 3.2]{Hrr08}, van der
Hoeven gave the first sub-quadratic bound to evaluate $d$ points
with modulus less than $1$. More
precisely, he showed that if $\|f\|_1<1$ it is possible to evaluate the $f(x_k)$ with an
error less than $2^{-m}$ with $\widetilde O(d^{3/2}m^{3/2})$ bit
operations. He achieved this bound by subdividing the unit disk in
annuli of constant width. A drawback of this bound is that it increases the
exponent on the required precision $m$.  By contrast, in our data structure,
we subdivide the unit disk with disks of width exponentially smaller
near the unit circle. This approach allows us to finally derive an
algorithm that is both quasi-linear in $d$ and in $m$ for the numerical
multipoint evaluation problem.

\begin{Theorem}
  \label{thm:evaluation}
  Given a polynomial $f$ of degree $d$ with $\|f\|_1 \leq 2^{\tau}$,
  Algorithm~\ref{alg:evaluation} returns the evaluation of $f$ on $d$ points in
  the unit disk with an absolute error less than $\|f\|_12^{-m}$
  in $\widetilde O(d(\tau+m))$ bit operations. 
\end{Theorem}
\begin{Remark}
  Using Algorithm~\ref{alg:evaluation} on the reverse polynomial
  $X^df(1/X)$, we see that the same complexity bound holds for a
  set of $d$ points in the complex plane, replacing $f$ by the function
  $$\widetilde f(x) = \begin{cases} f(x) & \text{if $x \in D(0,1)$}\\f(x)/x^d &
  \text{otherwise}\end{cases}.$$
\end{Remark}

Our algorithm is particularly well-suited for evaluating polynomials of
high degree with fixed constant precision such as machine precision.
This arises in many applications, notably in polynomial root
approximation. One of the most famous method to approximate a root of a
polynomial is the Newton method. Starting from an initial point $x_0$,
it consists in computing iteratively the Newton map, yielding the
sequence $x_{n+1} = x_n - f(x_n)/f'(x_n)$.  It was shown
that starting from an explicit set of $3.33 d \log^2 d$ points, this
approach is guaranteed to approximate all the roots of a given
polynomial \cite{HSSim01,BASmc16}. In their work, the authors only show experiences with
polynomial given by recursive formula, such that they can be evaluated
with a number of
operations logarithmic in their degree. Our data structure can
improve their approach for dense polynomials given by their list of
coefficients. In Section~\ref{sec:introrootfinding}, we focus on the computation of
disks isolating the roots of $f$, that is the computation of a set of
disks that are pairwise disjoint and that contain a unique root of $f$
each.

\paragraph{Main idea of the proof of Theorem~\ref{thm:evaluation}.}
Once we have computed a $m$-hyperbolic approximation of the input
polynomial $f$, the algorithm is rather straight-forward. For each disk
of the corresponding hyperbolic covering, we can efficiently find all
the input points that it contains, using a geometric data-structure such
as range searching, recalled in Section~\ref{sec:range}. Then, using
the state-of-the-art algorithm for multipoint evaluation
(Proposition~\ref{pro:evaluation}), we can evaluate the corresponding
approximate polynomial $g$ from the hyperbolic approximation on the selected points.
Since $g$ has a degree
lower than $m$, evaluating $g$ on the $x_i$ up to precision $m$ can be
done with an amortized time $\widetilde O(m)$ per point if the number of points is
greater than $m$ and in a total time $\widetilde O(m^2)$ if there are
less than $m$ points. Since the number of
disks in a hyperbolic approximation is in $O(d/m)$, this leads to a
bit complexity in $\widetilde O(dm)$. More details can be found in
Section~\ref{sec:evaluation}.

\subsubsection{Root isolation}
\label{sec:introrootfinding}
% Melhorn bound on root isolation https://doi.org/10.1016/j.jsc.2014.02.001
% Elias bound on the condition number \cite{https://link.springer.com/chapter/10.1007%2F978-3-319-24021-3_16}
% Bound on intersection searching Agarwal in quasi-linear time, section 3.6,  http://dx.doi.org/10.1090/conm/223

% TODO: say Melhorn requires multipoint evaluation on d points to
% guarantee roots in Lemma 8 or step 4
% 
% TODO: mention application of multipoint evaluation
% 
% TODO: add comments in algo
% TODO: describe algo, in particular: polynomials at most degree $d$,
% pairwise distinct boxes can be computed in quasilinear time with
% geometric data structures, Kantorovich theorems for root isolation 
% 
% TODO: say that there is no lower bound for the problem, yet all
% state-of-the art algo use at least $d$ evaluations of the input
% polynomials, or require to manipulate polynomials of degree $d$ with
% coefficients of bitsize $d$.

% TODO: say difficulties come from roots near the unit circle: use of
% Graeffe iteration to put the roots far away from the unit circle. Bit
% complexity quadratic Malajovich Zubelli.

Our data structure also allows us to improve the state-of-the-art
bound on the bit complexity for the problem of isolating all the complex roots of a given
polynomial,  with a bound that is adaptive in the condition number of the
input. The condition number of the zeros of a polynomial measures the
displacement of its roots with respect to a perturbation of its
coefficients (Definition~\ref{def:condition}). For a square-free
polynomial $f$ and a root $\zeta$, we let $\kappa_\zeta =
\frac{\max(1,|\zeta|^d)}{|f'(\zeta)|}$. Then the absolute condition
number of $f$ is $\kappa = \max_{f(\zeta)=0}(\kappa_\zeta)$.
Our goal is to provide for each root an isolating disk, that is a disk
that contains a unique root and doesn't intersect any other isolating
disk.
In our case we consider so-called \emph{projective disks}, that is
either a disk or the inverse of a disk $D$ defined as the
set of points $1/x$ such that $x \in D$ and $x \neq 0$.

% TODO: speak about generalized disks inverse of disks

\begin{Theorem}
  \label{thm:rootfinding}
  Given a square-free polynomial $f$ of degree $d$ with $\|f\|_1 \leq
  2^{\tau}$, with absolute condition number $\kappa$,
  Algorithm~\ref{alg:rootfinding} computes isolating projective disks
  of all the roots of $f$ in $\widetilde O (d (\tau + \log \kappa))$ bit
  operations.
\end{Theorem}
%\begin{Remark}
%  We can isolate the roots of $f$ outside the unit disk within the same
%  bit-complexity using Algorithm~\ref{alg:rootfinding} on the reverse
%  polynomial $X^df(1/X)$, and taking the inverse of the isolating disks.
%  Removing the duplicate solutions using range searching algorithms (see
%  Section~\ref{sec:range}, this allows us to isolate all the roots of
%  $f$ in $\widetilde O(d(\tau+\log \kappa))$ bit operations.
%\end{Remark}
% \begin{Remark}
%   The required precision for the arithmetic operations is in $O(\log
%   \kappa+\log d)$.
% \end{Remark}
\begin{Remark}
  \label{rem:coefficients}
  If all the coefficients of $f$ are real numbers, then
  Algorithm~\ref{alg:rootfinding} can be slightly modified to return all
  the real roots of $f$ with the same bit-complexity, by
  selecting the isolating projective disks that intersects the real
  axis.
\end{Remark}

%Rump https://www.sciencedirect.com/science/article/pii/S0377042703003819?via%3Dihub

In recent works on complex root isolation \cite{MSWjsc15,BSSYjsc18}, the
best adaptive complexity bound, rewritten with our
notation, is in $\widetilde O\left(d \sum_{f(\zeta)=0} [\tau +
\max(1,\log(\kappa_\zeta)) + \max(1,\log(\sigma_\zeta^{-1}))]\right)$,
where $\sigma_\zeta = \min_{f(\eta)=0, \eta\neq\zeta}(|\zeta-\eta|)$. Our bound removes
the dependency in $\sigma_\zeta$, and replaces the sum by a max, which
improves the state-of-the-art bounds by a factor $d$ if the condition
numbers of the roots are logarithmic in $d$ or evenly distributed. In adaptive algorithms
that compute isolating disks, % the bottleneck step is the evaluation of
a criterion is used for early termination. It usually checks if a disk or a
rectangle contains a unique root of $f$. Some examples of
criteria are detailed in Section~\ref{sec:kantorovich}.
% Among the criteria used in the literature for root isolation, one may
% mention Pellet's test \cite{}Melhorn, Pellet's test combined with
% Graeffe iteration \cite{BSSYjsc18,IPYicms18}, Cauchy's integral theorem
% \cite{IPissac20, IPmacis20}, Interval Newton \cite[c]{}, Neumaier, and others \cite{}, Rump, \ldots
All those
criteria end up evaluating a polynomial of degree roughly $d$ on
each of the $d$ roots (\cite[\S 2.2.3]{MSWjsc15},\cite[Lemma 5]{BSSYjsc18},
\cite[Algorithms~2 and~3]{IPissac20}, \ldots), which leads to a bit-complexity
at least quadratic in $d$ if we use a naive evaluation algorithm, or in
$d^{3/2}$ using state-of-the-art multipoint
evaluation methods \cite{Hrr08}. In our case, thanks to our new data
structure, the criterion to check that a disk contains a unique root is
replaced by the evaluation of a polynomial of degree $O(\tau +
\log(\kappa))$, which explains partly how we avoid a cost quadratic in $d$. 

% Mahler https://projecteuclid.org/journals/michigan-mathematical-journal/volume-11/issue-3/An-inequality-for-the-discriminant-of-a-polynomial/10.1307/mmj/1028999140.full

Another approach in the literature consists in computing the roots up to a precision high
enough such that we can guarantee that all the roots are
approximated correctly. An explicit bound exists in the case where the
input polynomial has integer coefficients. Schönhage showed \cite[\S
20]{Srr82}) that if we can compute $\widetilde f$ an approximate factorization
of $f$ close enough, then the roots of $f$ are
isolated by disks centered on the
roots of $\widetilde f$ with a radius depending on the condition
number $\kappa$.
% such that $\|f-\widetilde f\|_1 \leq
% 2^{-m}\|f\|_1$ for $m\geq \log_2(8d^2\|f\|_1^2\kappa^2)$ and the
% distance between a root $z$ of $\widetilde f$ and the closest root
% $\zeta$ of $f$ is less than $r=\frac 1 {4d^2\|f\|_1\kappa}$, then the
% disk $D(z, r)$ isolate the root $\zeta$ of $f$.
Combined with a bound on
$\kappa$ from Mahler \cite[last inequality]{Mmmj64} for polynomials with integer
coefficients, and using the algorithm of Pan \cite{Pjsc02} to compute
the approximate factorization $\widetilde f$, this leads to a root-isolating algorithm in
$\widetilde O(d^2\tau)$ bit operations \cite[\S 10.3.1]{EPTbook14}. Note that this
algorithm requires $d$ arithmetic operations performed with a precision in
$\Omega(d\tau)$, and requires at least a quadratic number of bit
operations. This method was also improved in practice for small degree
polynomials \cite{Grr93}.

In our case, the bound from Mahler on $\kappa$ implies that the bit complexity of
Algorithm~\ref{alg:rootfinding} is also in $\widetilde O(d^2\tau)$,
matching the bit-complexity of state-of-the art algorithms in the worst
case for square-free polynomials with integer coefficients.

\paragraph{Main idea of the proof of Theorem~\ref{thm:rootfinding}.}
Our approach roughly follows the original approach of
Schönhage \cite[\S 20]{Srr82} in the sense that we will compute isolating disks centered on
the roots of an approximation of $f$. It is also adaptive like more
recent works (\cite{MSWjsc15,BSSYjsc18,IPissac20}, among
others), in the sense that we use a criterion for early termination.

The main novelty of our approach is that we start by computing an
$m$-hyperbolic approximation of $f$, for a small initial constant
integer $m$. This returns a set of $O(d/m)$ polynomials $g$ of degree $O(m)$ defined on $O(d/m)$ disks
covering the unit disk $D$. Then we compute the approximate roots of
the $g$ and we use a criterion to check if the corresponding
approximate roots are the centers of disks isolating the roots of $f$.
If we did not find all the roots of $f$, we double the parameter $m$ and
we start again. The criterion that we use to check if an approximate
root is the center of an isolating disk is based on Kantorovich theory
(recalled in Section~\ref{sec:kantorovich})
and it can be tested using the approximate polynomials of degree lower
than $m$ coming from the hyperbolic approximation. Moreover,
this criterion will be satisfied for all roots of $f$ for
$m \geq c\log(\|f\|_1d\kappa)$ for a universal constant $c$
(Lemma~\ref{lem:termination}). Further details can be found in Section~\ref{sec:rootfinding}.

\subsubsection{Lower bound on condition number}
\label{sec:introcondition}
% Demmel 1987 : https://doi.org/10.1007/BF01400115
% Elias bound on the condition number \cite{https://link.springer.com/chapter/10.1007%2F978-3-319-24021-3_16}

%TODO: cite Wilkinson and say that we finally give a geometric
%explanation on why Wilkinson polynomial is ill-conditioned.

Another insightful application of our data structure is a new geometrical
interpretation of ill-conditioned polynomials, based on the distribution
of their roots. Since the introduction of Wilkinson's polynomials $p(X) =
(X-1)\cdots(X-d)$, it is known that the problem of polynomial
root finding can be ill-conditioned even in the cases where the roots are
well separated \cite{Wnm59, Wbook64}. More recently, it has been
proved that the condition
number of characteristic polynomials of $d \times d$ Gaussian matrices
is in $2^{\Omega(d)}$ in average \cite{Bjc17}, where
a Gaussian matrix is a matrix where the entries are independent,
centered Gaussian random variables.
Yet, no approaches provide a geometric explanation of this phenomenon. Our
next theorem provides a geometric criterion that allows one to detect
easily if a polynomial is ill-conditioned, based on the repartition of
its roots. In these works, the authors consider the relative condition
number defined as $\kappa^r(f) =
\max_{f(\zeta)=0}\left(\frac{\|f\|_1}{|\zeta|}\kappa_\zeta\right)$. 
%Trefethen and Bau remark that "Polynomial rootfinding is typically ill-conditioned even in cases that do not
%involve multiple roots.". 

\begin{Theorem}
  \label{thm:condition}
  Given a polynomial of degree $d$, let $N = \lceil \log_2(3ed) \rceil$
  and let $m$ be the maximal number of
  roots of $f$ (resp. $X^df(1/X)$) in a disk of the $N$-hyperbolic covering $\mathcal H_N$. The
  relative condition number $\kappa^r(f)$ of $f$ is greater than $\frac 1
  {4ed\sqrt{m}}2^{5m/11}$.
\end{Theorem}
\begin{Remark}
  In particular, the disk $D(0,1/2)$ is covered by $4$ disks in any
  $N$-hyperbolic covering. Thus, if $m$ is the number of roots with
  absolute value less than $1/2$ or greater than $2$, then
$\kappa^r(f) \geq \frac 1 {8ed\sqrt{2m}} 2^{5m/88}$.
\end{Remark}

As a direct consequence of this theorem, we recover the fact that the
Wilkinson's polynomials have a condition number in
$2^{\Omega(d)}$, since almost all their roots have a modulus larger than
$2$. Moreover, the set of eigenvalues of $d \times d$ Gaussian matrices
are the Ginibre determinantal point process \cite{Gjmp65, HKPVps06},
with eigenvalues roughly spread uniformly in the disk of radius
$\sqrt{d}$ centered at $0$, such that again, almost all the roots of the characteristic
polynomial have a modulus greater than $2$, which allows us to conclude
that the expectation of the logarithm of its condition number is in
$\Omega(d)$.

On the other hand, for a polynomial of degree $d$ with random coefficients following a centered,
Gaussian law of variance one, the expectation of the logarithm of the condition
number of its real roots in $O(\log d)$ \cite{DNVlms15}. This is
consistent with Theorem~\ref{thm:condition} since the roots of
such polynomial, called Kac polynomials or hyperbolic polynomials, are
roughly distributed evenly among the disks of an hyperbolic covering
\cite[\ldots]{EKbams95, STijm04, PVam05,KZap14}. This means that our root solver algorithm is
well-suited for polynomials with random coefficients of the same order
of magnitude.

\paragraph{Main idea for the proof of Theorem~\ref{thm:condition}.}
Given a polynomial $f$, we use Kantorovich theory to prove that for an integer
$m$ logarithmic in the condition number, the $m$-hyperbolic
approximation returns polynomials of degree $m-1$ that have at least as
many roots as $f$ in the corresponding disk of the hyperbolic covering.
Thus it implies that $m-1$ is greater than the number of roots of $f$,
which provides a lower bound on a quantity logarithmic in the condition number.

%Essentially, it proves that if the roots are not distributed
%roughly with respect to a hyperbolic metric, then it is
%ill-conditioned. TODO: precise the claim

% TODO: citation zeroes distribution Tsirelson, Sodin, https://doi.org/10.1007/BF02984409
% and Edelman-Kostlan
% and Peres-Virag
% https://projecteuclid.org/journals/acta-mathematica/volume-194/issue-1/Zeros-of-the-iid-Gaussian-power-series--a-conformally/10.1007/BF02392515.full
% 
% TODO: explain how it is a consequence of the fact that for a degree
% greater than a quantity depending on $\kappa$, the polynomials $g_k$
% have more roots than $f$ in the corresponding disk.
% 
% It allows us to recover
% previous lower bounds on the condition number characteristic
% polynomials, on Wilkinson polynomials, charcteristic polynomial, and so on.
% 
% TODO: insert figure of the roots of the Wilkinson polynomial, along with
% the disks, for N=10.

\subsubsection{Experimental proof of concept}
\label{sec:introexperiment}
% TODO: put this section in introduction.
% 
% TODO: mention application of root solving with random coefficients for
% generating repulsive point process
% 
% TODO: add a curve for m=40 and say validated. And mention that the
% Kantorovich criterion was validated in this case (check). No time !

Finally, we conclude with an experimental section, and we present a
simple implementation of a root solver in the programming language
\verb/Python/, using the standard numerical library \verb/NumPy/
\cite{numpy}, and
working with machine precision. Since our implementation is less than
$100$ lines of code, we include it in Appendix. 

Our implementation is a simplified version of
Algorithm~\ref{alg:rootfinding}. For the approximate factorization and
the fast evaluation of the roots of unity, we use the standard polynomial root solver
and the Fast Fourier Transform procedures of
\verb/NumPy/. For the data structure to detect duplicates, we simply round
the solutions to a lower precision and sort the rounded solutions to
detect the values with the same binary representation after rounding.

% TODO: mention that we don't do checking in this implementation and
% results might be wrong for ill-conditioned polynomials. Nevertheless in
% our experiments, we always found all the roots, matching those of
% mpsolve with an error less than 1e-8 (check).

The current state-of-the-art implementation of a root solver for complex polynomials is the
software \verb/MPSolve/ \cite{BAna00, BRjcam14}, implemented in the \verb/C/ programming language,
and based notably on the Aberth-Ehrlich method \cite{Eacm67}
. Its development started more than $20$ years
ago and it has received several improvements over time, making it the
fastest current implementation to find all the complex roots of a
polynomial. This software also uses multi-precision arithmetic
when necessary. By contrast, our solver \verb/HCRoots/ is an early prototype written in
\verb/Python/, working in machine precision
only, and depending solely on the \verb/NumPy/ library. Nevertheless, as
we can see in Figure~\ref{fig:benchmark}, for random polynomial that are
known to be well-conditioned (see
Section~\ref{sec:introcondition}), our solver \verb/HCRoots/
called with a precision parameter $m=30$ is an order of
magnitude faster than \verb/MPSolve/, which is very promising.

In our experiments, we focused on polynomials where the coefficients
are centered, Gaussian random variables with variance $1$. In this case,
the solutions returned by our solver matched all the solutions returned by
\verb/MPSolve/ with an error less than $2^{-25}$, for polynomials up to
degree $25000$. Moreover, the linear
complexity of our algorithm, combined with the fact that we don't need
to use multi-precision arithmetic, allowed us to solve polynomials of degree
25000 an order of magnitude faster than \verb/MPSolve/. In
Figure~\ref{fig:benchmark}, we show the timings to solve random
polynomials of degree $d$ with our solver, and \verb/MPSolve/. Moreover,
it is easy to parallelize our algorithm, and improve furthermore its
practical efficiency.

% TODO: describe a bit the implementation versus the theoretical
% algorithm.

%TODO: add figure benchmark against MPSolve with simple numerical implementation in python for root approximation.

\begin{figure}
  \centering
  \includegraphics[width=0.7\linewidth]{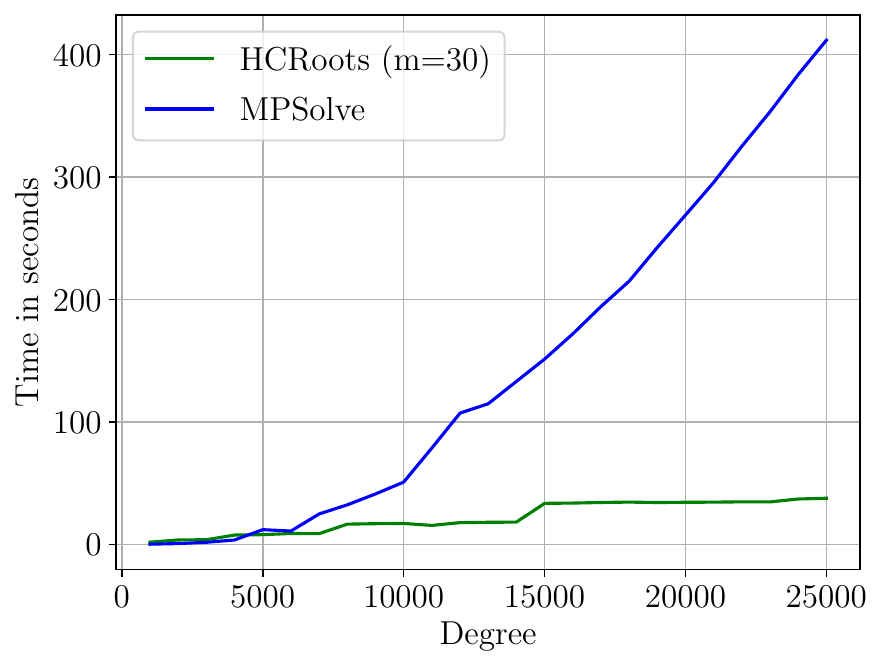}
  \caption{Time to approximate the roots of polynomials where the
  coefficients are random variable, centered Gaussian of variance $1$}
  \label{fig:benchmark}
\end{figure}

% \subsubsection{Other problems}
% TODO: so far we only covered some applications, we expect more
% application of our data structure, notably for multivariate case. Also
% for generation of random points with repulsive law.

% TODO: say that it is beyond the scope of this article

\section{Preliminaries}
\subsection{Notations}

Given a polynomial or an analytic series $f$, we will denote by $f'$ and
$f''$ the derivative and the second derivative of $f$. If $f$ is a
polynomial, we will denote by $\|f\|_1$, $\|f\|_2$ and $\|f\|_\infty$
the classical norm $1$, $2$ and infinity on the vector of its
coefficients. We will denote by $D(\gamma, \rho)$ the complex
disk of radius $\rho$ centered at $\gamma$. Finally, for a polynomial
$f$ and a root $\zeta$ of $f$, we let $\kappa_\zeta =
\frac{\max(1,|\zeta|^d)}{|f'(\zeta)|}$ and $\kappa_1(f) =
\max_{f(\zeta)=0}(\kappa_\zeta)$ is the absolute condition number of
$f$. It is also denoted $\kappa(f)$ and referred as the condition number of $f$.

% TODO: Define norm 1 of polynomials
% TODO: Recall $\kappa$ notations. Derivative notations. Disk notation.
% TODO: define $\varTheta$, $\widetilde O$.
% 
% TODO: notation for hyperbolic approximation.
% 
% TODO: change all theorems into proposition in preliminary section.
% 
% TODO: change $z_k$ to $\zeta_k$ for homogeneity.

\subsection{Fast elementary operations}
\label{sec:elementary}

We start with classical results on elementary operations, notably on the
multiplication and the composition of polynomials.

\begin{Proposition}[{\cite[Theorem 2.2]{Sca82}, \cite[Proposition~3.2]{Rtcs86}}]
  \label{pro:multiplication}
  Let $f$ and $g$ be two polynomials of degree $d$ with $\|f\|_1$ and
  $\|g\|_1$ less than $2^{\tau}$, and an integer $m \geq \log (d+1)$. It is possible to
  compute the polynomial $h$ such that $\|h-fg\|_1\leq 2^{-m}$ in
  $\widetilde O(d(m+\tau))$ bit operations.
\end{Proposition}

Using Fast Fourier transform algorithm, we can also evaluate in a
quasi-linear time a polynomial on the roots of unity. Note that in this
case, even if the required precision $m$ is smaller than $d$, the
algorithm is still quasi-linear in $dm$.

\begin{Proposition}[{\cite[\S 3]{Sca82}, \cite[Proposition~3.3]{Rtcs86}}]
  \label{pro:fft}
  Let $f$  be a polynomial of degree $d$, and $\|f\|_1 \leq 2^{\tau}$ with $\tau \geq 1$,
  and an integer $m \geq \log (d+1)$. It is possible to compute the
  complex numbers $y_0, \ldots, y_{d-1}$ such that $\sum_{k=0}^{d-1}|y_k-f\left(e^{i2\pi
  k/d}\right)|\leq 2^{-m}$ in $\widetilde O(d(m+\tau))$ bit operations.
\end{Proposition}

Finally, another classical result that we will use is the fast
composition of polynomials.

\begin{Proposition}[{\cite[Theorem 2.2]{Rtcs86}}]
  \label{pro:composition}
  Let $f$ and $g$ be two polynomials of degree $d$ with $\|f\|_1 \leq
  2^\tau$ and $\|g\|_1 \leq 2^\nu$ where $\tau\geq1$ and $\nu \geq 1$.
  Let $m$ be a positive integer. It is possible to compute the
  polynomial $h$ of degree $d-1$ such that $\|h(X)- f(g(X)) \mod X^d\|_1\leq 2^{-m}$
  in $\widetilde O(d(m + \tau +  d\nu)))$ bit operations.
\end{Proposition}
\begin{Remark}
  Ritzmann \cite[Theorem 2.2]{Rtcs86} used the same bound for $\|f\|_1$ and $\|g\|_1$. Our
  proposition is a direct consequence of Ritzmann's theorem if we
  multiply $f$ in the input by $2^{\nu-\tau}$ , and the result by
  $2^{\tau-\nu}$. This reduction can be done in $\widetilde
  O(d(\tau+\nu+m))$ bit operations.
\end{Remark}

\subsection{Fast approximate multipoint evaluation}

% TODO: add complexity for fast fourrier transform and multiplication of
% polynomials from Schönhage
% TODO: add complexity for fast composition (Ritzmann https://www.sciencedirect.com/science/article/pii/0304397586901076 )
% 
% 
% TODO: cite Proposition 13 in $d^{3/2}$ van der hoeven https://www.texmacs.org/joris/fastcomp/fastcomp-abs.html
% TODO: say that they come closer using a partition of the disks with
% annuli of same size. What makes our approach work is that our disks
% cover annuli of size exponentially smaller when closer to the unit
% circle.

When we want to evaluate a polynomial on multiple points, it is possible
to amortize the number of bit operations when the precision required is
larger than the degree. In a recent work \cite[\S 3.2]{Hrr08}, van der
Hoeven showed that it can be done in $\widetilde O(d^{3/2}m^{3/2})$.
However, this bound is not optimal when $m$ is greater than $d$.
For the case $m>d$, we recall here another state-of-the-art bound on the bit complexity for fast multipoint
evaluation.

\begin{Proposition}[{\cite[Lemma 11]{Hrr08}, \cite[Theorem~9]{KSarxiv13}}]
  \label{pro:evaluation}
  Let $f$ be a polynomial of degree $d$, with $\|f\|_1 \leq 2^\tau$, with
  $\tau \geq 1$, and let $x_1, \ldots, x_d \in \mathbb C$ be complex
  points with absolute values bounded by $1$. Then, computing $y_k$ such
  that $|y_k - f(x_k)| \leq 2^{-m}$ for all $k$ is possible in
  $\widetilde O(d(m + \tau + d))$ bit operations.
\end{Proposition}

% Note that even if all the points have an absolute value less than $1$,
% their approach require to take $\nu \geq 1$, such that this method
% is quadratic in the degree.

Even though the bit complexity is quadratic in $d$, this approach is
near optimal when $m$ is greater than $d$, since its complexity matches
the size of the output in this case. We reuse notably this result to
bound the complexity of Algorithm~\ref{alg:evaluation}, since our
approach reduces the problem of evaluating a polynomial of
degree $d$ to the problem of evaluating several polynomials of degree
$m$ with a precision greater than $m$.

%Using our data structure, we improve this result in the case where
%the points we want to evaluate are in the unit disk.

%Theorem of fast approximate multipoint evaluation
\subsection{Condition number}
\label{sec:preliminarycondition}
Our root isolation algorithm has a bit complexity that depends on the
condition number of the input polynomial. The absolute condition number
is a measure of the displacement of its roots with respect to the
displacement of its coefficients. More precisely, for a polynomial $f$
with a vector of coefficients $c \in \mathbb C^{d+1}$, and a root
$\zeta$ of $f$, there exists a neighborhood $U \subset \mathbb C^{d+1}$
of $c$, a neighborhood $V \subset \mathbb C$ of $\zeta$ and a
differentiable function $\psi: U \rightarrow V$ that maps $c \in U$ to
the unique zero in $V$ of the corresponding polynomial. Letting
$D\psi(\zeta)$ be the gradient of $\psi$ at $\zeta$, the condition
number of $\psi$ at $\zeta$ is the induced norm $|\|D\psi(\zeta)\||_2 =
\max_{\|\delta\|_2 = 1} |D\psi(\zeta) \cdot \delta| = \left(\sum_{k=0}^d
|\zeta|^{2k}\right)^{1/2}$. If we consider the induced norm $1$ instead,
we have $|\|D\psi(\zeta)\||_1 =
\max_{\|\delta\|_1 = 1} |D\psi(\zeta) \cdot \delta| = \max_{k=0}^d(
|\zeta|^{k}) = \max(1, |\zeta|^d)$.

%vector $\delta \in \mathbb C^{d+1}$ of perturbation values with
%$\|\delta\|_1 = 1$, we look at the function 

\begin{Definition}\cite[\S 14.1.1]{Bbook13}
  \label{def:condition}
  The \emph{standard local absolute condition number} of polynomial $f$ of
  degree $d$ at a root $\zeta$ is $\kappa_2(f, \zeta) = \frac 1 {|f'(\zeta)|} \left(\sum_{k=0}^d
  |\zeta|^{2k}\right)^{1/2}$. Considering all the roots, we define the
  \emph{standard
  absolute condition number} of $f$ as
  $\kappa_2(f) = \max_{f(\zeta)=0} \kappa_2(f, \zeta)$.
\end{Definition}
\begin{Remark}
  \label{rem:norm}
  The standard absolute condition number is obtained by considering the norm $2$.
  If we consider the induced norm $1$ instead, we get the condition
  number $\kappa_1(f,\zeta) = \frac{\max(1, |\zeta|^d)}{|f'(\zeta)|}$, such
  that $\kappa_1(f) \leq \kappa_2(f) \leq \sqrt d \kappa_1(f)$.
\end{Remark}

As shown in Remark~\ref{rem:norm}, the bound depending on the logarithm
of the condition number for the norm $1$ and the norm $2$ will be the
same up to a factor logarithmic in $d$.  In the following, we will focus
on the condition number $\kappa_1$ induced by the norm $1$.

For square-free polynomial with integer
coefficients, the condition number is finite. The following proposition
bounds the condition number for square-free polynomials with integer
coefficients.

%TODO: replace this proposition by the Mahler bound

\begin{Proposition}[{\cite[last inequality]{Mmmj64}}]
  Given a square-free polynomial $f$ of degree $d$ with integer coefficients, let $\tau$ be a real
  such that $\|f\|_1 \leq 2^\tau$. Then $\log(\kappa_1(f))$ is in
  $O(d\tau + d\log d)$.
\end{Proposition}

In particular, combined with Theorem~\ref{thm:rootfinding}, this
proposition implies that for square-free polynomials, the bit-complexity of Algorithm~\ref{alg:rootfinding} has the same
worst-case bound as the state-of-the-art root-finding methods.
% On the other hand, for random polynomials, TODO: mention condition
% number logarithmic.
% We can also bound the condition number with respect to the separation
% bound since $\kappa_1(f,\zeta) =\max(1,\zeta^d)/\prod_{f(\eta)=0}
% (\zeta-\eta)$. This allows us in particular to recover and improve the
% bounds on root isolation depending on separation bounds (\cite{},
% Melhorn).
We define also the relative condition number to represent the
relative displacement of the roots, with respect to a relative
displacement of the coefficients.

\begin{Definition}[{\cite[\S 14.1.1]{Bbook13}}]
  \sloppy The \emph{standard local relative condition number} of a polynomial
  $f$ at a root $\zeta \neq 0$ is
  $\kappa^r_2(f,\zeta) = \frac{\|f\|_2}{|\zeta|}\kappa_2(f, \zeta)$, and the
  \emph{standard relative condition number} of $f$ is
  $\kappa^r_2(f) = \max_{f(\zeta)=0, \zeta \neq 0} \kappa^r_2(f, \zeta)$. Similarly, we
  define $\kappa^r_1(f,\zeta) = \frac{\|f\|_1}{|\zeta|}\kappa_1(f, \zeta)$, and the
  associated \emph{relative condition number} of $f$ is
  $\kappa^r_1(f) = \max_{f(\zeta)=0, \zeta \neq 0} \kappa^r_1(f, \zeta)$.
\end{Definition}
\begin{Remark}
  \label{rem:transpose}
  An interesting property of the relative condition number at a nonzero
  root $\zeta$ of $f$ is that it is equal to the relative condition
  number at the root $1/\zeta$ of the polynomial
  %identical for the polynomial $f(X)$ of degree $d$ and the polynomial
  $g(X) = X^d f(1/X)$. Moreover for the absolute condition number we
  have $\kappa_1(g, 1/\zeta) \leq \kappa_1(f,\zeta)$ for $\zeta$
  outside the unit disk. 
\end{Remark}
\begin{proof}[Proof of the remark]
By construction, the inverse function is a bijection between the non-zero roots $\mu$ of $g$
and the non-zero roots $\zeta$ of $f$. Computing
the derivative of $g$ at a root $\mu$ we have
$\mu g'(\mu) = d\mu^df(1/\mu) - \mu^{d-1} f'(1/\mu) = - \zeta
f'(\zeta)/\zeta^d$. Thus $\kappa_1(g,\mu)/|\mu| =
\frac{|\zeta|^d \max(1,|\mu|)^d%\left(\sum_{k=0}^d |\mu|^{2k}\right)^{1/2}
  }{|\zeta|\cdot|f'(\zeta)|}= \frac{\max(1,|\zeta|)^d%\left(\sum_{k=0}^d |\zeta|^{2k}\right)^{1/2}
}{|\zeta|\cdot|f'(\zeta)|} =
\kappa_1(f,\zeta)/|\zeta|$. Finally, since $\|f\|_1 = \|g\|_1$ by
construction, we have $\kappa_1^r(g,\mu) = \kappa_1^r(f,\zeta)$.
\end{proof}

This number is a standard way to measure to stability of the roots with
respect to independent perturbations of the coefficients. In
Theorem~\ref{thm:condition}, we provide a geometric criterion to bound
from below this condition number.

% Note that for all roots with absolute value greater than
% $1$, the relative condition number $\kappa_2^r(f, \zeta)$ is smaller
% than $\|f\|_2\kappa_2(f,\zeta)$. In particular for polynomials with
% all roots of absolute value greater than $1$,
% Theorem~\ref{thm:condition} provides also a lower bound for the absolute
% condition number.

\subsection{Fast approximate factorization}

Another important result on univariate polynomials is a bound on the bit
complexity to approximate all its roots. In the complex, approximating
the roots is equivalent to compute an approximate factorization. We
recall the state-of-the-art bound on the bit complexity for this problem.

\begin{Proposition}[{\cite[Theorem 2.1.1]{Pjsc02}}]
  \label{pro:pan}
  Let $f$ be a polynomial of degree $d$ with leading coefficient $c_d$ and
  all its roots $\zeta_k$ in the unit
  disk, and $m \geq d \log d$ a fixed real number. It is possible to
  compute complex numbers $z_1, \ldots, z_d$ such that $\|f(X) -
  c_d\prod_{k=1}^d(X-z_k)\|_1 \leq 2^{-m}\|f\|_1$ in $\widetilde O(d m)$
  bit operations.
\end{Proposition}
\begin{Remark}
  \label{rem:pan}
  The theorem also holds with the same complexity for a polynomial $h$
  that has all its roots in the disk centered at the origin and of
  radius $c2^{m/d}$ for a constant $c \geq 1$, such as in the polynomials $h$ in
  Algorithm~\ref{alg:rootfinding} (Lemma~\ref{lem:compact}).
\end{Remark}
\begin{proof}[Proof of the remark]
  Let $f(Y) = h(c2^{m/d}Y)$. Then $f$ has all its
  roots in the unit disk and $\|f\|_1 \leq
  c^d2^m\|h\|_1$. Computing the approximate factorization $\widetilde
  f$ of $f$ such that $\|f-\widetilde f\|_1 \leq 2^{-2m-d\log_2
  c}\|f\|_1$ can be done in $\widetilde O(d m)$ since $m\geq
  d\log d$. Then with the change of variable $Y = X 2^{-m/d}/c$, we have
  $\|h(X) - \widetilde f(X 2^{-m/d}/c)\|_1 \leq 2^{-m} \|h\|_1$.
\end{proof}

This theorem does not directly give a bound on the distances between the
roots. Schönhage shows \cite[\S 19]{Srr82} that in the worst case
$|\zeta_k - z_k| < 4 \cdot 2^{-m/d}$. This bound can be improved for
well-conditioned roots. In our case, we will also need a bound on the
distance between some pairs of roots of two polynomials that have
different degrees. We will use Kantorovich theory for the bounds in these
cases (Section~\ref{sec:kantorovich} and~\ref{sec:approximation}).

% For
% completeness, we prove in Lemma~\ref{lem:rootapproximation} that if the
% condition number of $f$ at a root $z_k$ is $\kappa$ and $m >
% \log_2(d^2\kappa)$, then
% under the condition of Proposition~\ref{pro:pan}, the approximate root
% $z_k$ satisfies $|\zeta_k - z_k| < \frac 1 {d^2\kappa}$.
% TODO: update precise bounds according to Lemma

Note that Remark~\ref{rem:pan} requires a bound on the modulus of
the roots of the polynomial. For that, we will use the classical
Fujiwara bound.

% TODO: replace F by f when necessary
% Fujiwara https://www.jstage.jst.go.jp/article/tmj1911/10/0/10_0_167/_article

\begin{Proposition}[{\cite{Ftmj16}}]
  \label{pro:fujiwara}
  Let $f=\sum_{k=0}^d f_k X^k$ be a polynomial of degree $d$. Then the
  moduli of the roots of $f$ are lower or equal to $2\max_{1\leq k\leq
d} \sqrt[k]{\left|\frac{f_{d-k}}{f_d}\right|}$
\end{Proposition}

Applying this bound on the polynomial returned by an hyperbolic
approximation, slightly perturbed, will allow us to verify that the assumptions of
Remark~\ref{rem:pan} are satisfied (Lemma~\ref{lem:compact}).

\subsection{Certification of the roots}
\label{sec:kantorovich}

Several approaches in the literature guarantee that a neighborhood of a
point contains a unique root of a given polynomial. We may cite notably
Kantorovich criterion \cite[\S 3.2]{Dbook06}, Smale's alpha theorem
\cite[\S 3.3]{Dbook06}, Newton interval method \cite[Theorem
5.1.7]{Nbook90}, Pellet's test \cite{MSWjsc15}, Pellet's test combined
with Graeffe iteration \cite{BSSYjsc18,IPYicms18}, Cauchy's integral
theorem \cite{IPissac20, IPmacis20}, and others \cite{Rjcam03}, \ldots 
%among others.
A first crude bound from the literature to guarantee that a disk centered at a point $x$ contains
a root of a polynomial $f$ is the following.

\begin{Proposition}[{\cite[Theorem 6.4e]{Hbook74}, \cite[Theorem
  9]{BAna00}}]
  \label{pro:crude}
  Let $f$ be a polynomial of degree $d$, and $x$ a complex point. Let
  $r_k = \sqrt[k]{k! \binom{d}{k}\left|\frac{f(x)}{f^{(k)}(x)}\right|}$. Then, for
  all $1\leq k \leq d$, the disk $D(x, r_k)$ contains a root of $f$. In
  particular, the disk $D(x, d |f(x)/f'(x)|)$ contains a root of $f$.
\end{Proposition}

%In Algorithm~\ref{alg:rootfinding}, 

With some additional conditions, Kantorovich criterion provides a
smaller radius to guarantee that a disk contains a root $f$.

% to find
% isolation disks. We recall here the results that we use, specifically
% for the case of a complex univariate function.

\begin{Proposition}[{\cite[Theorem 88]{Dbook06}}]
  \label{pro:existence}
  Given a function $f$ of class $C^2$, and a point $x$ such that
  $f'(x) \neq 0$, let $\beta = |f(x)/f'(x)|$, and $K =
  \sup_{|y-x|\leq 2\beta}|f''(y)/f'(x)|$. If $2\beta K \leq 1$, then
  $f$ has a root in the disk $D(x, 2|f(x)/f'(x)|)$.
\end{Proposition}

The previous propositions are useful to find a disk that contains a root
of $f$, but they don't guarantee that the disk contains a unique root of
$f$.  In order to prove that Algorithm~\ref{alg:rootfinding} terminates and to
bound its complexity, we use the lower bound from Kantorovich theory on
the size of the basin of attraction of the roots of $f$. More precisely,
for each root $\zeta$, we bound the size of a disk containing $\zeta$
where the Newton method always converges toward $\zeta$.

\begin{Proposition}[{\cite[Theorem 85]{Dbook06}}]
  \label{pro:unicity}
  Given a function $f$ of class $C^2$, and a root $\zeta$ such that
  $f'(\zeta) \neq 0$. If there exists $r>0$ such that $2rK \leq 1$, with $K =
  \sup_{|x-\zeta|\leq r}|f''(x)/f'(\zeta)|$, then $\zeta$ is the unique
  root of $f$ in the disk $D(\zeta, r)$. Moreover, for any $x_0 \in
  D(\zeta, r)$, the Newton sequence defined by $x_{n+1} = x_n -
  f(x_n)/f'(x_n)$ converges toward $\zeta$.
\end{Proposition}

\begin{Remark}
  \label{rem:distinct}
  If $f$ is a polynomial of degree $d$, and  $s \geq \sup_{x\in
D(0,1)} |f''(x)|$, then
%is bounded by $d^2\|f\|_1$ for $x$ in the unit disk. In particular,
a consequence of Proposition~\ref{pro:unicity} is that the set of disks $D(\zeta,
1/(2s\kappa_1(f)))$, for all roots $\zeta$ of $f$ in the unit
disk, are pairwise distinct.
\end{Remark}

For any complex point $x$, we also prove the following lemma to have a
criterion guaranteeing that a ball around $x$ is included in a basin
of attraction of a root of $f$.

\begin{Lemma}
  \label{lem:distinct}
  Given a function of class $C^2$ and a point $x \in \mathbb C$ such
  that $f(x) \neq 0$. Let $r > 2 |f(x)/f'(x)|$ and $K > |f''(y)/f'(x)|$ for
  all $y \in D(x, 4r)$. If $5rK \leq 1$ then, $f$ has a unique root $\zeta$
  in $D(x, r)$, and for all $x_0 \in D(x,r)$, the Newton sequence
  starting from $x_0$ converges to $\zeta$.
\end{Lemma}
\begin{proof}
  According to Proposition~\ref{pro:existence}, $f$ has a root $\zeta$
  in $D(x,r)$. Then, $|f'(\zeta)| \geq |f'(x)| - |x-\zeta|\sup_{y\in
  D(x,r)}|f''(y)|$. This implies that $f'(\zeta) \geq |f'(x)|(1-rK) \geq 4|f'(x)|/5$. Letting
  $\widetilde K = 5K/4$, we have that $\widetilde K \geq
  |f''(y)/f'(\zeta)|$ for all $y \in D(x, 4r) \supset D(\zeta, 2r)$. Remark
  that $4 r \widetilde K \leq 1$, such that using
  Proposition~\ref{pro:unicity}, this implies that for all $x_0 \in
  D(\zeta, 2r) \supset D(x,r)$, the Newton sequence starting from $x_0$
  converges toward $\zeta$.

\end{proof}

\subsection{Geometric range searching}
\label{sec:range}

Our data structure can be seen as a piecewise polynomial approximation.
As such, when doing multipoint evaluation, we will need to report the points
that fall in a disk. This problem can be solved efficiently using
classical range searching and point intersection searching algorithms.

% reduced to report points
% that fall in a halfspace in $\mathbb R^3$ and can be done efficiently.

\begin{Proposition}[{\cite[\S 5.2, Table 7]{AEcm99}}]
  \label{pro:disks}
  Given $n$ points $x_i$ in $\mathbb C$, it is possible to compute a data
  structure in $\widetilde O(n)$ operations such that for any disk $D$,
  returning the list of points $x_i$ contained in $D$ can be done in $O(k + \log
  n)$ operations, where $k$ is the number of points in $D$.
\end{Proposition}

Moreover, when we isolate the roots of a polynomial $f$, we reduce the
problem to isolate the roots in each disk of an $N$-hyperbolic covering.
Because those disks overlap, we need to remove redundant boxes. For
that, we use fast rectangle-rectangle searching techniques.

\begin{Proposition}[{\cite[\S 3.6]{AEcm99}}]
  \label{pro:rectangles}
  Given $n$ rectangles $r_i$ in the plane, it is possible to compute a data
  structure in $\widetilde O(n)$ such that for any rectangle $r$,
  returning the list of rectangle $r_i$ intersecting $r$ can be done in
  $\widetilde O(k + \log n)$ operations, where $k$ is the number of
  rectangles intersecting $r$.
\end{Proposition}

Note that in an $N$-hyperbolic covering, each disk intersect at most $10$
other disks of the covering. Indeed the disks centered in the disk
$D(0,\frac 1 2)$, it intersects at most $10$ other disks. For $n\geq1$ and a disk with a center between the circle
of radius $r_n$ and a circle of radius $r_{n+1}$, it intersects
$2$ disks of the hyperbolic covering that have their centers in the same
ring. Then it can intersects at most $2$ other disks coming from the
inner adjacent ring, and $4$ other from the outer adjacent ring. Thus,
in Algorithm~\ref{alg:rootfinding}, this will guarantee that each query
will be done in $\widetilde O(\log n)$ (see
Section~\ref{sec:proofrootfinding}).

% result that we both use and improve in
% Algorithm~\ref{alg:rootfinding} is the 
% 
% Theorem of fast approximate factorization
% 
% Theorems of Kantorovich based on bound on the second derivative

\section{Computation of the hyperbolic approximation}
\label{sec:datastructure}

% TODO: check page number overlapping algorithm layout
% 
% TODO: change comments to emph and with \#
% 
% TODO: add python code in appendix

\begin{algorithm}
    \DontPrintSemicolon
    \KwInput{A polynomial $f(X) = \sum_{k=0}^d f_kX^k$ of degree $d$
    with $\|f\|_1 \leq 2^\tau$, $\tau\geq 1$,\\ \hspace{1.35cm}and an
  integer $m \geq 1$}
    \KwOutput{An $m$-hyperbolic approximation of $f$ (see
    Definition~\ref{def:approximation})}
    %A list of pairs $(g_k, a_k)$ where:
    %  \begin{itemize}
    %    \item $g_k$ is a polynomial of degree $\min(d, \log(1/\varepsilon))$
    %    \item $a_k$ is an affine map of the form $a_k(x) =
    %      (\gamma_k + \rho_k x) e^{i\alpha_k}$\\
    %          with $\gamma_k, \rho_k$ positive real numbers and $\alpha_k \in [0,2\pi[$
    %\end{itemize}
    %and such that for all $x \in D(0,1)$ there exists $k$ such that $x
    %\in a_k(D(0,1))$ and
    %  $$|g_k(a_k^{-1}(x)) - f(x)| < \|f\|_1\varepsilon$$
    %TODO: rewrite output
    %}
    %$m \gets \log_2(1/\varepsilon)$\;
    $\widetilde m \gets \min(m-1, d)$\;
    $N \gets \lceil\log_2(3ed/\widetilde m)\rceil$\;
    \For{$n$ from $0$ to $N-1$}{
      \tcc{Compute $(g_{n,k}, a_{n,k})$ for the disks covering $D(0,r_{n+1}) \setminus D(0, r_n)$}
      \tcc{The precision of the arithmetic operations is in
      $\varTheta(m+\tau+\log d)$}%\\ TODO: check if $\log d$ necessary too}
      \tcc{A. Compute $r_n, \gamma_n, \rho_n$ and $K_n$ for the $a_{n,k}(X) = (\gamma_n + \rho_nX)e^{i2\pi \frac k {K_n}}$}
      $r_n \gets 1 - 1/2^n$\;
      $r_{n+1} \gets 1 - 1/2^{n+1}$ if $n \leq N-2$ else $1$\;
      $\gamma_n \gets (r_{n} + r_{n+1})/2$\;
      $\rho_n \gets \frac 3 4(r_{n+1} - r_n)$\;
      %$\alpha_n \gets 2 \arccos\left(\frac{\gamma_{n}^2 + r_{n+1}^2 - \rho_n^2}{2\gamma_n r_{n+1}} \right)\text{ for } 0 \leq n \leq N-1$\;
      $K_n \gets \lceil \frac{3\pi}{\sqrt 5}
      \frac{r_{n+1}}{\rho_n}\rceil$\vspace{1em}\;
      %\tcp*{$K_n$ is in $\varTheta(2^n)$ (see Lemma~\ref{lem:bla})}\;
      %TODO: maybe remove truncation to $d_n$ since exponent on log d is not tracked anymore\;
      %TODO: Use polynomial composition of Ritzmann instead\;
      %Write $p(YZ) \mod Z^{K_n}-1$ 
      \tcc{B. Compute $g_{n,k}(X) \approx f\left((\gamma_n + \rho_n X)e^{i2\pi \frac k {K_n}}\right) \mod X^m$}
      \tcc{B.1. Truncate $f$ at $d_n$ such that $(\gamma_n+\rho_n)^{d_n+1} \leq 1/2^{m+1}$}
      $d_n \gets \min\left(d,\lceil  \frac 8 3 \log(2) (m+1) 2^n \rceil-1\right)$ if $n<N-1$ else $d$\;
      $p \gets f_0 + \cdots + f_{d_n} X^{d_n}$\;
      \tcc{B.2. Gather the coefficients in $Y$ of $p(YZ) \mod
      Z^{K_n}-1$,\\\phantom{B.2. }where $Y$ and $Z$ are symbolic variables}
      %$\sum_{k=0}^{K_n-1} p_k(Y^{K_n})Y^kZ^k \gets p(YZ) \mod Z^{K_n}-1$ \tcp*{$p_\ell$ has degree in $O(m)$}
      %$p_0(Y^{K_n}) + \cdots + p_{K_n-1}(Y^{K_n})Y^{K_n-1}Z^{K_n-1} \gets p(YZ) \mod Z^{K_n}-1$
      %$\sum_{k=0}^{K_n-1} p_k(Y^{K_n})Y^kZ^k \gets p(YZ) \mod Z^{K_n}-1$ \tcp*{$p_\ell$ has degree in $O(m)$}
      \For{$k$ from $0$ to $K_n-1$}{
        $p_k(Y^{K_n})Y^k \gets$ coefficients of $Z^k$ of $p(YZ) \mod Z^{K_n}-1$\;
      }
      \tcc{B.3. Compute $(\gamma_n + \rho_n X)^k \mod X^{\widetilde m}$}
      $q_0(X) \gets 1$\;
      \For{$k$ from $1$ to $K_n$}{
        $q_k(X) \gets q_{k-1}(X) \cdot (\gamma_n + \rho_n X) \mod
        X^{\widetilde m}$\;
      }
      \tcc{B.4. Compute $r_k(X) = p_k\left((\gamma_n + \rho_n
      X)^{K_n}\right) \cdot (\gamma_n + \rho_n X)^k \mod X^{\widetilde
  m}$}
      \For{$k$ from $0$ to $K_n-1$}{
        $r_{k,0} + \cdots + r_{k,{\widetilde m}-1}X^{{\widetilde m}-1}
        \gets p_k(q_{K_n}(X)) \cdot q_k(X) \mod X^{\widetilde m}$ \;
      }
      \tcc{B.5. Compute $g_{n,k}(X) = r_0(X) + \cdots + r_{K_n-1}(X)
      e^{i2\pi\frac k {K_n}(K_n-1)}$}
      \For{$\ell$ from $0$ to ${\widetilde m}-1$}{
        $s_\ell(Z) \gets r_{0,\ell} + \cdots + r_{K_n-1, \ell}Z^{K_n-1}$\;
        $g_{n,0,\ell}, \ldots, g_{n,K_n-1,\ell} \gets s_\ell(e^{i2\pi
        \frac 0 {K_n}}), \ldots, s_\ell(e^{i2\pi \frac {K_n-1} {K_n}})$\;
      }
      \tcc{B.6. Append the pair to the result list}
      \For{$k$ from $0$ to $K_n-1$}{
        $g_{n,k}(X) \gets$ $g_{n,k,0} + \cdots + g_{n,k,{\widetilde
        m}-1}X^{{\widetilde m}-1}$\;
        $a_{n,k}(X) \gets (\gamma_n + \rho_n X) e^{i2\pi \frac k {K_n}}$\;
        Append the pair $(g_{n,k}, a_{n,k})$ to the list $L$\;
      }
      % \;\;
      % \tcc{Computing $v_{k,\ell} = p\left((\gamma_n + \rho_ne^{i2\pi \frac k m})e^{i2\pi \frac \ell {K_n}}\right)$}
      % Rewrite $p$ as $p_0(X^{K_n}) + \cdots +
      % p_{K_n-1}(X^{K_n})X^{K_n-1}$ \tcp*{$p_\ell$ has degree in
      % $O(m)$}\;
      % \For{$k$ from $0$ to $m-1$}{
      %   \For{$\ell$ from $0$ to $K_n$}{
      %     $z_{k,\ell} \gets (\gamma_n + \rho_ne^{i2\pi \frac k m})^\ell$\;
      %   }
      % }
      % \For{$\ell$ from $0$ to $K_n-1$}{
      %   $y_{0, \ell}, \ldots, y_{m-1, \ell} \gets p_\ell(z_{0,K_n}),
      %   \ldots, p_\ell(z_{m-1, K_n})$\;
      % }
      % \For{$k$ from $0$ to $m-1$}{
      %   \tcc{Compute the polynomial $q_k$ such that $q_k(e^{i\frac
      %   {2\pi \ell} {K_n}})$ $ = p\left((\gamma_n + \rho_ne^{i\frac {2\pi k} m})e^{i\frac{2\pi\ell} {K_n}}\right)$}
      %   $q_k(Y) \gets y_{k, 0}z_{k,0} + \cdots +
      %   y_{k, K_n-1}z_{k,K_n-1}Y^{K_n-1}$ \;
      % $v_{k, 0}, \ldots, v_{k, K_n-1} \gets q_k(e^{i2\pi \frac 0 {K_n}}),
      %   \ldots, q_k(e^{i2\pi \frac {K_n-1} {K_n}})$\;
      % }
      % \tcc{Interpolating}
      % \For{$\ell$ from $0$ to $K_n-1$}{
      %   $g_\ell(X) \gets$ Polynomial interpolation on the point-value pairs $\left(e^{i2\pi \frac k m}, v_{k,\ell}\right)$\;
      %   $a_\ell(X) \gets (\gamma_n + \rho_n X) e^{i2\pi \frac \ell {K_n}}$\;
      %   Append the pair $(g_\ell, a_\ell)$ to the list $L$\;
      % }

      % TODO: maybe mention Taylor shift in comments
    }
    \Return $L$\;

  \caption{Hyperbolic approximation data structure}
  \label{alg:datastructure}
\end{algorithm}

\subsection{Properties}
First, given an integer $m$ and a polynomial $f$ of degree $d$, and a
pair $(g,a)$ from the hyperbolic approximation $H_{d,m}(f)$, we give
a bound on the coefficients of the polynomials $f(a(X))$. That gives also a bound on
the polynomial $g$, since $g$ is an approximation of the polynomial
$f(a(X))$ truncated to the degree $m-1$.

\begin{Lemma}
  \label{lem:coefficients}
  Given a polynomial $f$ of degree $d$ and a integer $m>1$, let
  $a$ be an affine transformation appearing in the
  hyperbolic approximation $H_{d,m}(f)$. Letting $f(a(X)) =
  \sum_{k=0}^{d} c_k X^k$ and $\widetilde m = \min(m-1,d)$,
  %for all $0\leq k \leq \widetilde m$
  we have $|c_k| \leq \begin{cases} \|f\|_1\left(\frac {\widetilde m}{2k}\right)^k
  & \text{ if } 0 \leq k \leq \widetilde m\\
  \|f\|_1/2^k & \text{ otherwise}\end{cases}$
\end{Lemma}
\begin{proof}
  By construction, $a$ is of the form $a(X) = (\gamma + \rho X)e^{i2\pi
  \alpha}$, with $\gamma$ and $\rho$ two positive real numbers. Letting
  $f(X) = \sum_{k=0}^d f_k X^k$, and expanding the $a(X)^k$, we have
  $|c_k| \leq \sum_{\ell=k}^d |f_\ell|
  \binom{\ell}{k} \gamma^{\ell-k} \rho^k$. We now distinguish two
  cases. First if $\gamma + \rho < 1$, then by construction of the
  sequences in Definition~\ref{def:covering}, we have $1 - \gamma = 2
  \rho$. Thus $|c_k| \leq \frac 1 {2^k}\sum_{\ell=k}^d |f_\ell|
  \binom{\ell}{k} \gamma^{\ell-k} (1-\gamma)^k \leq \|f\|_1/2^k$. This
  proves the desired bound for both the cases where $0\leq k
  \leq \widetilde m$ and the case where $k \geq \widetilde m$.
  
  Then, if $\gamma+\rho > 1$, then $\gamma = 1 - 1/2^N$ and $\rho =
  3/2^{N+1}$, where $N = \lceil \log_2(3ed/\widetilde m) \rceil$. Then
  $|c_k| \leq \sum_{\ell=k}^d |f_\ell| \binom{\ell}{k} \gamma^{\ell-k}
  \rho^k \leq \|f\|_1\rho^k \binom{d}{k}$. Using the inequalities $k! >
  (k/e)^k$ and $d \cdots (d-k+1) \leq d^k$ we have $|c_k| \leq \|f\|_1
  (\rho e d / k)^k$. Moreover, $\rho \leq \widetilde m / (2de)$, such
  that $|c_k| \leq \|f\|_1 (\widetilde m / (2k))^k$. This also
  proves the desired bound for both the cases where $0\leq k
  \leq \widetilde m$ and the case where $k \geq \widetilde m$.
\end{proof}

We also give a bound on the number of disks appearing in the
decomposition.

\begin{Lemma}
  \label{lem:disks}
  Given two integers $d$ and $m>1$, let $\widetilde m = \min(m-1,d)$ and let $N=\lceil \log_2(3ed/\widetilde
  m) \rceil$. Then the number of disks in the $N$-hyperbolic covering is
  in $O(d/\widetilde m)$. Moreover, the union of the disks contains the
  unit disk.
  %with the notations of Definition~\ref{def:covering}, for $0 \leq n \leq N-1$, we have
  %$2^n\leq K_n \leq 2^{n+4}$.
\end{Lemma}
\begin{proof}
  %TODO: update proof
  First, remark that the total number $t$ of disks in
  a $N$-covering is $\sum_{n=0}^{N-1} K_n$. Such that using
  Remark~\ref{thm:datastructure} we have $K_n \leq
  2^{n+4}$, and $t \leq 2^{N+4} \leq 16 \cdot 3ed/\widetilde m$. Thus
  the number of disks is in $O(d/\widetilde m)$.

  Then, we need to prove that for any ring $R_n = D(0,r_{n+1}) \setminus
  D(0,r_{n})$, the union of the disks centered at $\gamma_ne^{i2\pi\frac
  k{K_n}}$ with radius $\rho_n$ contains $R_n$. For that, let $D$ be the
  disk $D(\gamma_n, \rho_n)$ and let $R_n(\alpha)$ be the segment intersection
  of $R_n$ with the half-line starting from $0$ and with angle $\alpha$.
  Then, let $\beta$ be the smallest angle such that $R_n(\beta)$ is not
  included in $D$. Then if the angle $\frac {2\pi}{K_n} \leq 2 \beta$ it
  implies that for all $0 \leq \alpha < 2\pi$, the segment $R_n(\alpha)$
  is included in a disk of the $N$-hyperbolic covering.

  Since $\sin(\beta) \leq \beta$, it is sufficient to prove that $\frac
  \pi {K_n} \leq \sin(\beta)$. Consider the triangle $ABC$ with $A=\gamma_n$, $B=r_{n+1}$ and
  $C=r_{n+1}e^{i\beta}$ (see Figure~\ref{fig:geometry}). Letting $O=0$, remark that $\frac \beta 2$ is the
  angle $\angle BOC$.

  Let $h$ be the distance between the point
  $C$ and the line $(OB)$. By construction, $\sin(\beta) =
  \frac h {r_{n+1}}$. If we prove that the angles at vertex $A$ and
  vertex $B$ are each smaller than $\frac \pi 2$, then we can conclude
  that $h^2 \geq |C-A|^2 -
  |B-A|^2 = \rho_n^2 - (r_{n+1}-\gamma_n)^2 = \rho_n^2 - \frac 4
  9 \rho_n^2 = \frac 5 9 \rho_n^2$. Such that $\sin(\beta) \geq
  \frac {\sqrt 5}{3} \frac {\rho_n}{r_{n+1}} \geq \frac \pi {K_n}$.
                                                                                                
  In order to prove that the angles at $A$ and $B$ are smaller than
  $\frac \pi 2$, remark that the triangle $OBC$ is isosceles, such that
  the angle $\angle OBC = \angle ABC$ is less than $\pi/2$.  Moreover,
  considering the triangle $OAC$, the angle $\angle OAC$ is larger than
  $\pi/2$ if $|C-O|^2 - |A-O|^2 - |C-A|^2 \geq 0$. This inequality holds
  for $n \geq 1$ since $|C-O|^2 - |A-O|^2 - |C-A|^2 = r_{n+1}^2 -
  \gamma_n^2 - \rho_n^2 = (r_{n+1}-\gamma_n)(r_{n+1}+\gamma_n) -
  \rho_n^2 \geq \frac 2 3 \rho_n - \rho_n^2 \geq 0$. Thus $\angle BAC =
  \pi - \angle OAC$ is less than $\frac \pi 2$.  Since the angle of the
  triangle $ABC$ at the vertices $B$ and $C$ are less than $\frac \pi
  2$, we can conclude that for all $n$,  $\frac \pi {K_n}$ is small
  enough to let the disks cover the unit disk.

\end{proof}

\begin{figure}
  \centering
  \includegraphics{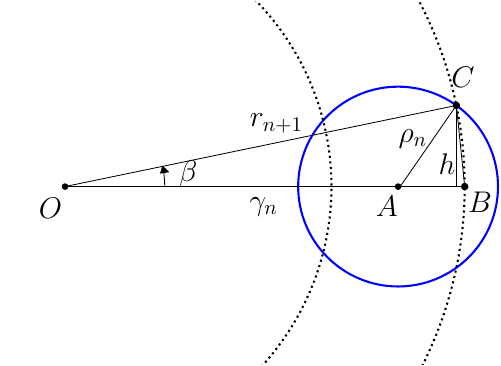}
  \caption{Illustration for the proof that the union of the disks in a $N$-hyperbolic
  covering contains the unit disk (Lemma~\ref{lem:disks})}
  \label{fig:geometry}
\end{figure}

Finally, we give a bound on the bit size of the coefficients of the
polynomials $g$ of an hyperbolic approximation, by bounding the size of
the polynomials $f(a(X))$, leading also to a bound the its second
derivative.

\begin{Lemma}
  \label{lem:second}
  Given a polynomial $f = \sum_{k=0}^d f_k X^k$ of degree $d$ and $a$ an affine transform
  from an $m$-hyperbolic approximation of $f$, let $\varphi(X) =
  f(a(X))$. Letting $\widetilde m = \min(m-1, d)$, we have $\|\varphi\|_1
  \leq \|f\|_1 2^{\widetilde m/11}$ and for all $x \in D(0,1)$, we have
  %$|\varphi'(x)| \leq \|f\|_1\widetilde m2^{\widetilde m/11}$ and
  $|\varphi''(x)| \leq \|f\|_1\widetilde m^22^{\widetilde m/11}$.
\end{Lemma}
\begin{proof}
  Let $a$ be of the form $a(X) = (\gamma + \rho X)e^{i2\pi
  \alpha}$. Let $f^+(X) = \sum_{k=0}^d |f_k| X^k$. Since $\gamma$
  and $\rho$ are positive numbers, we have $\|\varphi\|_1 \leq f^+(a(1)) \leq
  f^+(\gamma+\rho)$. By construction, $\gamma + \rho < 1 + \frac 1 4
  \frac 1 {2^{N-1}}$ with $N \geq \log_2(3ed/{\widetilde m})$. This implies that
  $\gamma+\rho \leq 1+ \frac {\widetilde m} {6ed}$ and thus $\|\varphi\|_1 \leq
  \|f\|_1(1+\frac {\widetilde m} {6ed})^d \leq \|f\|_1 e^{\frac {\widetilde m} {6e}}
  \leq \|f\|_1 2^{{\widetilde m}/11}$.
  Then, we obtain the bound on the derivative and second derivative
  of $\varphi$ by bounding the
  absolute value of the derivative and second derivative of each term
  $f_ka(X)^k$. %by $|f_k| k(k-1) \rho^2(\gamma+\rho)^{k-2}$.
  By construction, either
  $\gamma + 2\rho \leq 1$ or $\rho = \frac 3 4 \frac 1 {2^{N-1}}$ and
  $\gamma + \rho = 1 + \frac 1 4 \frac 1 {2^{N-1}}$. 

  % For the derivative, in the former case,
  % using binomial inequalities, we have $|f_k| k
  % \rho(\gamma+\rho)^{k-1} \leq |f_k|(\gamma+2\rho)^k \leq |f_k|$, such
  % that $\varphi'(x) \leq \|f\|_1$. In this case, we also have for the
  % second derivative $|f_k| k(k-1)
  % \rho^2(\gamma+\rho)^{k-2} \leq 2|f_k|(\gamma+2\rho)^k \leq 2|f_k|$, such
  % that $\varphi''(x) \leq 2\|f\|_1$. 

  % In the second case, with $N \geq \log_2(3ed/{\widetilde m})$, we have
  % $|f_k| k \rho(\gamma+\rho)^{k-1} \leq |f_k| d \rho(\gamma+\rho)^{d-1}
  % \leq |f_k|d \frac{{\widetilde m}}{2ed}\left(1+\frac {\widetilde
  %     m} {6ed}\right)^{d-1} \leq |f_k|\frac{\widetilde m}{2e}2^{{\widetilde
  % m}/11}$, such that $|\varphi'(x)| \leq \|f\|_1{\widetilde
  % m}2^{{\widetilde m}/11}$.  Similarly, for the second derivative we
  % have $|f_k| k(k-1) \rho^2(\gamma+\rho)^{k-2} \leq |f_k| d(d-1)
  % \rho^2(\gamma+\rho)^{d-2} \leq |f_k|d(d-1) \left(\frac{{\widetilde
  % m}}{2ed}\right)^2\left(1+\frac {\widetilde m} {6ed}\right)^{d-2} \leq
  % |f_k|{\widetilde m}^2/(4e^2)2^{{\widetilde m}/11}$, which implies
  % $|\varphi''(x)| \leq \|f\|_1\widetilde m^22^{\widetilde m/11}$.

  In the former case,
  using binomial inequalities, we have $|f_k| k(k-1)
  \rho^2(\gamma+\rho)^{k-2} \leq 2|f_k|(\gamma+2\rho)^k \leq 2|f_k|$
   such that $|\varphi''(x)| \leq 2\|f\|_1$. 
  In the second case, with $N \geq \log_2(3ed/{\widetilde m})$, we have
  $|f_k| k(k-1) \rho^2(\gamma+\rho)^{k-2} \leq |f_k|
  d(d-1) \rho^2(\gamma+\rho)^{d-2} \leq |f_k|d(d-1)
  \left(\frac{{\widetilde m}}{2ed}\right)^2\left(1+\frac {\widetilde m} {6ed}\right)^{d-2} \leq
  |f_k|{\widetilde m}^2/(4e^2)2^{{\widetilde m}/11}$, which implies
  $|\varphi''(x)| \leq \|f\|_1\widetilde m^22^{\widetilde m/11}$.

  % With $N \geq \log_2(3ed/m)$, we have
  % $\gamma + 2\rho \leq 1+\frac 1 {2^{N-1}} \leq 1 +\frac {2m}
  % {3ed}$. Rising this expression to the power $d$ gives us the bound
  % $\varphi''(x) \leq e^{2m/(3e)} \leq 2\|f\|_1 2^{m/2}$.

\end{proof}

%Algo

%Proof of that hyperbolic covering covers the unit disk

\subsection{Proof of Theorem~\ref{thm:datastructure}}

We can now prove that Algorithm~\ref{alg:datastructure} returns an
hyperbolic approximation of a polynomial with a complexity quasi-linear
in the degree and the required precision.

\paragraph{Correctness.}
First, for the correctness of the algorithm, we will prove that
Algorithm~\ref{alg:datastructure} returns a list of pairs $(g,a)$
satisfying the constraints of Definition~\ref{def:approximation}. First
the affine transforms computed in Algorithm~\ref{alg:datastructure} send
the unit disk to the disks described in Definition~\ref{def:covering} of a hyperbolic
covering. Then, the polynomials $g$ computed are approximation of the
polynomials $f(a(X)) \mod X^{\widetilde m}$. We will show that the
approximation satisfies the bound $\|g(x) - f(a_{n,k}(X))\|_1 \leq 3\|f\|_12^{-m}$.

In part $B.1$ of Algorithm~\ref{alg:datastructure}, we start by
truncating $f$ to $d_n$. The resulting polynomial $p$ satisfies
$f-p = f_{d_n+1}X^{d_n+1} + \cdots + f_dX^d$. 
Let $a_{n,k}$ be of the form $a_{n,k}(X) = (\gamma_n + \rho_n X)e^{i2\pi
\frac k K_n}$ and let $e^+(X) = \sum_{k=d_n+1}^d |f_k| X^k$. Since
$\gamma_n$
and $\rho_n$ are positive numbers, we have $\|f(a_{n,k}(X)) -
p(a_{n,k}(X))\|_1 \leq \|e^+(a_{n,k}(X))\|_1 \leq
e^+(\gamma+\rho)$. In the case where $n<N-1$ and $d_n < d$, we have $\gamma_n + \rho_n = 1
-\frac 3 8 \frac 1 {2^n} < 1$. Thus $\|e^+(a_{n,k}(X))\|_1 \leq
(\gamma_n+\rho_n)^{d_n+1} \|f\|_1$. Moreover, $(\gamma_n+\rho_n)^{d_n+1}
\leq e^{(d_n+1)\log(1-\frac 3 8 \frac 1 {2^n})} \leq e^{-(d_n+1)\frac 3 8
\frac 1 {2^n}} \leq \frac 1 {2^{m+1}}$ since $d_n \geq \frac 8 3
\log(2)(m+1)2^n$. This implies that $\|f(a_{n,k}(X))-p(a_{n,k}(X))\|_1 \leq \|f\|/2^{m+1}$.

Then, the algorithm will evaluate $p(YZ)$ on $Y = \gamma_n + \rho_nX$ 
and $Z = e^{i2\pi\frac k {K_n}}$, modulo $X^{\widetilde m}$ and
modulo $Z^{K_n}-1$. The advantage is that composition modulo
$X^{\widetilde m}$ can be done efficiently using
Proposition~\ref{pro:composition} and evaluation modulo $Z^{K_n}-1$ on
$Z=e^{i2\pi\frac k {K_n}}$ can
be reduced to Fast Fourier Transform and be done efficiently too using
Proposition~\ref{pro:fft}. The part $B.2$ up to $B.5$ can perform this
computation with an error less than $\|f\|_1/2^{m+1+\log_2(\widetilde
m)}$ on each coefficients. More precisely, in $B.3$ the algorithm
computes $q_k(X) = (\gamma_n+\rho_n X)^k \mod X^{\widetilde m}$. Then
in $B.4$ we have $r_k(X) = p_k((\gamma_n+\rho_n X)^{K_n}) \cdot
(\gamma_n + \rho_n X)^k \mod X^{\widetilde m}$. Letting $\omega =
e^{i2\pi / K_n}$, the truncated
polynomial $g_{n,k}$ associated to the disk of center $\gamma_n
\omega_k$ and radius $\rho_n$ is $\sum_{j=0}^{K_n-1} r_j(X) \omega^{kj}$.
In step $B.5$, for a fixed $n$ and a fixed $\ell$ between $0$ and
$\widetilde m - 1$, the coefficient of $X^\ell$ of $g_{n,k}$, denoted by
$g_{n,k,\ell}$, is $s_\ell(\omega^k)$. Those coefficients can be
computed efficiently using the fast Fourier transform algorithm.
Such that using the notation of
Algorithm~\ref{alg:datastructure} we finally have $\|g_{n,k}(X) - p(a_{n,k}(X))
\mod X^{\widetilde m}\|_1 \leq \|f\|_1/2^{m+1}$.

% Finally, using
% Lemma~\ref{lem:coefficients} on $p$, we have that $\|p(a(X)) - [p(a(X))
% \mod X^{\widetilde m}]\|_1 \leq 

Finally, letting $p(a_{n,k}(X)) = \sum_{k=0}^d c_kX^k$, we have by
Lemma~\ref{lem:coefficients} that $|c_k|\leq \|p\|_1/2^k$ for all
$k\geq\widetilde m$. Thus, we have $\|p(a_{n,k}(X)) - [p(a_{n,k}(X)) \mod
X^{\widetilde m}]\|_1 \leq 2 \|p\|_1/2^{\widetilde m} \leq 2 \|f\|_1/2^{\widetilde m} $.

Gathering the norm inequalities, we have as required:
\begin{align*}
  \|f(a_{n,k}(X)) - g(X)\|_1 & \leq \|f(a_{n,k}(x)) - p(a_{n,k}(X))\|_1
  + \|p(a_{n,k}(X)) - [p(a_{n,k}(X)) \mod X^{\widetilde m}]\|_1 \\
                             & \phantom{\leq} + \|p(a_{n,k}(X)) - g(X) \mod X^{\widetilde m}\|_1\\
                             & \leq \|f\|_1/2^{m+1} + \|f\|_1/2^{m+1} +
                             2\|f\|_1/2^m \\
                             &\leq 3\|f\|_1/2^m.
\end{align*}

% if we approximate the first coefficients
% $c_k$ for $k < \widetilde m$ by coefficients $g_k$ such that $|g_k-c_k|
% \leq \|f\|_1/(m2^m)$, we will have $\|g(X) - f(a(X))\|_1 \leq
% \|f\|_1/2^m + 2\|f\|_1/2^m = 3\|f\|_1/2^m$.

\paragraph{Complexity.}
The number of loop iterations in Algorithm~\ref{alg:datastructure} is in
$O(\log d)$. Thus, it is sufficient to prove that each iteration can be
performed in $\widetilde O(d(m+\tau))$ to achieve the complexity in
Theorem~\ref{thm:datastructure}. First, part $A$ can be done in
$\widetilde O(\log d)$ operations. Then in part $B$, we will use the
state-of-the-art complexity bounds on the elementary operations recalled
in Section~\ref{sec:elementary}.

First, in part $B.1$ and $B.2$, we are
reordering the coefficients, gathering together the
coefficients of $X^k$ with the same value $k \mod K_n$, which can be
done in $\widetilde O(d)$ bit operations. Note that the polynomials $p_k$
computed in this part have a degree less than $d_n/K_n$, with $d_n/K_n \leq \frac 8 3
\log(2)(m+1)$ and $d_n/K_n \leq d$, such that the degree of $p_k$ is in
$O(\widetilde m)$.

Then in part $B.3$, we do $K_n$ multiplication of polynomials of degree
in $\widetilde O(\widetilde m)$ with an absolute error on the result in $2^{\varTheta(m)}$.
Moreover, we have $\|q_k\|_1 \leq \max(1, (\gamma_n+\rho_n)^{K_n})$.
Since $\gamma_n+\rho_n \leq 1 + \frac 1 {2^{N+1}}$, this implies
$\|q_k\|_1 \leq e^{K_n / 2^{N+1}} \leq e^{2^{N+3}/2^{N+1}} \leq e^2$.
Using Proposition~\ref{pro:multiplication}, each multiplication can be
done in $\widetilde O(\widetilde m m)$ bit operations, and part $B.3$ requires
$\widetilde O(K_n\widetilde m m ) = \widetilde O(dm)$ bit operations.

Similarly, in part $B.4$, since $\|p_k\|_1 \leq \|p\|_1 \leq \|f\|_1
\leq 2^{\tau}$ and $\|q_{K_n}\|_1 \leq e^2$, using proposition
\ref{pro:composition} on fast composition,
we can compute $p_k(q_{K_n}) \mod X^{\widetilde m}$ with an error less
than $2^{-\varTheta(m)}$ in $\widetilde O(\widetilde m(m
+ \tau))$. We also perform the multiplication by $q_k$ in this part within the same
complexity. Overall, since the composition and the multiplication is
done $K_n$ times, part $B.4$ can be done in $\widetilde O(d(m+\tau))$
operations.

Finally, in part $B.5$, we use the Fast Fourier Transform algorithm to
evaluate $s$ of degree $K_n$ on the roots of unity $e^{i2\pi \frac k
{K_n}}$. If $\|s_{\ell}\|_1 \leq 2^{\nu}$ for an integer $\nu \geq 1$, the
evaluation of $s$ on the $K_n$ points with an error less than $2^{-m}$
can be done in $\widetilde O(K_n(m + \nu))$ using
Proposition~\ref{pro:fft}. Thus the bit complexity for part $B.5$ is in
$\widetilde O(d(m + \nu))$.  To bound $\|s_\ell\|_1$, remark that 
$\sum_{\ell=0}^{\widetilde m -1} \|s_\ell\|_1 \leq
\sum_{k=0}^{K_n-1} \|p_k(q_{K_n}))\cdot q_k\|_1$. Let $p_k^+, p^+, f^+$ be the
polynomial $p_k, p, f$ where we replaced the coefficients by their absolute
value. In this case $\|p_k(q_{K_n}))\cdot q_k\|_1 \leq
p_k^+((\gamma_n+\rho_n)^{K_n}) (\gamma_n+\rho_n)^k$.  Moreover
$\sum_{k=0}^{K_n-1} p_k^+((\gamma_n+\rho_n)^{K_n}) (\gamma_n+\rho_n)^k =
p^+(\gamma_n+\rho_n) \leq f^+(\gamma_n+\rho_n)$. In turn
$f^+(\gamma_n+\rho_n) \leq \|f\|_1 (1+\frac 1 {2^{N+1}})^d$, and $N\geq
\log_2(3ed/\widetilde m)$, such that $(1+\frac 1 {2^{N+1}})^d \leq
(1+\frac{\widetilde m}{6ed})^d \leq e^{\frac{\widetilde
m}{6e}}\leq 2^{\widetilde m/11}$. Finally, $\sum_{\ell=0}^{\widetilde
m-1} \|s_\ell\|_1 \leq \|f\|_12^{\widetilde m / 11} \leq 2^{\tau + m}$.
Thus part $B.5$ can be computed in $\widetilde O(d(m+\tau))$ bit
operations.

%TODO: gather propositions to conclude.

%\begin{proof}[Proof of Theorem~\ref{thm:datastructure}]
%
%\end{proof}

%Proof of correctness of the output approximation bound

% Proof of time complexity (split in several Lemma)

\section{Multipoint evaluation}
\label{sec:evaluation}

A direct application of our data structure is the fast evaluation of
polynomials. The main idea is to approximate the input polynomial $f$ with
a piecewise polynomial, where each polynomial $g_k$ has a degree with the same order of magnitude as the
required precision. Then we can use state-of-the-art multipoint
evaluation technique on each $g_k$.

\begin{algorithm}
  \DontPrintSemicolon
  \KwInput{Polynomial $f$ of degree $d$, $d$ complex number $x_i$ in the
  unit disk, and a precision $m$}
  \KwOutput{List of complex number $y_i$ such that $|y_i - f(x_i)| \leq \|f\|_1 2^{-m}$}
    $L \gets \{\}$\;
    $Q \gets$ data structure adapted to the $x_i$ for fast disk range searching \;
    % $\varepsilon \gets 4 \|f\|_1 / 2^{m}$\;
    % $\varepsilon' \gets 16\|f\|_1 (m+1)^2 / 2^{m}$\;
    $G \gets H_{d,m+2}(f)$ \;
    \For{$(g_k, a_k)$ in $G$}{
      %TODO: use $G$ also in root isolation algo\;
      %\tcc{Add lead monomial so that all its roots are in the disk of radius $e2^{max(1,m/d)}$}
      \tcc{The precision of the arithmetic operations is in
      $\varTheta(\tau+m)$}
      $v_1, \ldots, v_{n_k} \gets$ query $Q$ for list of points $x_i$ in $a_k(D(0,1))$\;
      $y_1, \ldots, y_{n_k} \gets g_k(a_k^{-1}(v_1)), \ldots , g_k(a_k^{-1}(v_{n_k}))$\;
      Append $y_1, \ldots, y_{n_k}$ to $L$\;
    }
    \Return $L$\;
  \caption{Multipoint evaluation}
  \label{alg:evaluation}
\end{algorithm}

\begin{proof}[Proof of Theorem~\ref{thm:evaluation}]
  The correction of Algorithm~\ref{alg:evaluation} is ensured by the
  fact that for all $x$ in a disk $a_k(D(0,1))$, letting $z =
  a_k^{-1}(x)$, we have $|f(x) - g_k(a_k^{-1}(x))| = |f(a_k(z)) -
  g_k(z)| \leq \|f(a_k(X)) - g_k(X)\|_1 \leq 3 \|f\|_12^{-m-2}$. If we
  compute $y$ the evaluation of $g_k(z)$ with an error less than $\varepsilon =
  \|g_k\|_12^{-12m/11-2}$, the result will have an error less than
  $\|f\|_12^{-m-2}$, using the bound on $\|g_k\|_1$ given in
  Lemma~\ref{lem:second}. So finally we have $|f(x) - y| \leq 3
  \|f\|_12^{-m-2} + \|f\|_12^{-m-2} = \|f\|_12^{-m}$.

  First the data structure $Q$ can be computed in $\widetilde O(d)$
  using Proposition~\ref{pro:disks}, and the hyperbolic approximation
  $G$ in $\widetilde O(dm)$ using Theorem~\ref{thm:datastructure}. Then in the
  loop, the algorithm queries the points $v_1, \ldots, v_{n_k}$ in
  $\widetilde O(n_k + \log d)$ using Proposition~\ref{pro:disks}.
  Let $\widetilde m = \min(m+1,d)$ be the degree of $g_k$, and $q_k =
  \lceil n_k/\widetilde m \rceil$. We can evaluate $g_k$ on $n_k$
  points using $q$ times the fast multipoint evaluation method in
  Proposition~\ref{pro:evaluation}. For an absolute error less than
  $\|g_k\|_12^{-12m/11-2}$, this can be done in $\widetilde O(q_k\widetilde m m)$ bit
  operations. Note that $q_k \widetilde m \leq n_k+\widetilde m$, such that the total
  complexity in an iteration of the \verb/for/ loop is in $\widetilde
  O(n_km + \widetilde m m + \log d)$. Note also that the sum of the $n_k$ is $d$.
  If $t$ is the number of discs in the hyperbolic approximation, after adding the complexity of all the
  main loop iterations, Algorithm~\ref{alg:evaluation} requires
  $\widetilde O(dm + t(\widetilde m m+\log d))$. By Lemma~\ref{lem:disks}, $t$ is in
  $O(d/\widetilde m)$, such that the total complexity of
  Algorithm~\ref{alg:evaluation} is in $\widetilde O(dm)$.
\end{proof}

%Proof of the corollary

\section{Root isolation}
\label{sec:rootfinding}

%TODO: replace $f$ by $\varphi$.

% TODO: handle more carefully the inverse polynomials

\begin{algorithm}
  \DontPrintSemicolon
  \KwInput{Squarefree polynomial $f$ of degree $d$}
  \KwOutput{List of $d$ disks isolating all the roots of $f$ in the unit
  disk and a subset of the other roots }
    % TODO: check constants \;
    $L \gets \{\}$\;
    $m \gets 1$\;
    \While{$|L| < d$}{
      % $\varepsilon \gets 4 \|f\|_1 / 2^{m}$\;
      % $\varepsilon' \gets 16\|f\|_1 (m+1)^2 / 2^{m}$\;
      $L \gets \{\}$\;
      $G \gets H_{d,m}(f)$ \;
      $G^* \gets \{(g,\frac 1 a) \mid (g,a) \in H_{d,m}(X^df(1/X))\}$\;
      \For{$(g, a)$ in $G \cup G^*$ \label{num:startfor}}{
        %\tcc{Add lead monomial so that all its roots are in the disk of radius $e2^{max(1,m/d)}$}
        %$\widetilde m \gets \min(2d, m)$\;
        \tcc{Reduce the upper bound on the disk containing the roots of $g$ (Lemma~\ref{lem:compact})}
        $\widetilde m \gets \min(m-1, d)$\;
        $h \gets g(X)+\frac{\|f\|_1} {2^{m}}X^{2\widetilde m}$\;
        \tcc{Compute an approximation of the roots of $h$}
        $\widetilde h \gets$ Approximate factorization of $h$ such that
        $\|h - \widetilde h\|_1 \leq 2^{-1.1m} \|h\|_1$\;
        \For{$z_j$ root of $\widetilde h$}{
          \tcc{Check root unicity of $f(a(X))$ in a neighborhood
          of $z_j$ (Lemma~\ref{lem:distinct})}
          $\varepsilon \gets 3 \|f\|_1 (m+2) / 2^{m}$ \tcp*{bound on $|f-g|$ and $|f'-g'|$ (Lemma~\ref{lem:derivative})}
          \If{$|z_j| \leq 1$ and $|g'(z_j)| > \varepsilon$}{
            $K \gets \frac{\|f\|_1{\widetilde m}^22^{\widetilde m/11}}{|g'(z_j)| - \varepsilon}$\;
            $\beta \gets \frac{|g(z_j)|+\varepsilon}{|g'(z_j)| -
            \varepsilon}$\;
            \If{$10\beta K \leq 1$ and $D(z_j, 8\beta) \subset D(0,1)$
              \label{step:criterion}}{
              $L \gets L \cup \{a(D(z_j, 2\beta))\}$\;
            }
          }
        }
      \label{num:endfor}}
      \tcc{Remove duplicate roots}
      $B \gets$ list of bounding box of disks in $L$\;
      $Q \gets$ data structure adapted to squares in $B$ optimized for rectangle-rectangle search\;
      $L \gets$ sublist of $L$ without duplicates\;
      $m \gets 2m$\;
    }

    \Return $L$\;
  \caption{Root isolation}
  \label{alg:rootfinding}
\end{algorithm}

%TODO: publish code on gitlab

%TODO: update proof algo 3, notably remove the underscore k

%TODO: Algo
%We also give a bound on the second derivative of the polynomials

% TODO: optional, move this lemma to section properties of data structure

% TODO: add subsections, one being properties of the polynomials of a hyperbolic approximation.

\subsection{Properties of the approximate roots}
\label{sec:approximation}

We start by describing the properties satisfied by the roots of the truncated polynomials $g$
coming from an $m$-hyperbolic approximations. In particular, we show how to
perturb them such that all their roots are contained in a small enough
disks.

%TODO: simplify statement with only one case: $1/2X^{2d}$.

% TODO: move Fujiwara bound in preliminary section

\begin{Lemma}
  \label{lem:compact}
  Let $m>\widetilde m$ be two positive integers. Let $g(X) = \sum_{k=0}^{\widetilde m}
  c_k X^k$ be a polynomial of degree $\widetilde m$
  and $c$ be a constant such that $c_0 \leq c$ and $|c_k| \leq c \left(\frac {\widetilde m}
  {2k}\right)^k$ for $k \geq 1$. Then for the roots of the polynomial $g(X) + \frac
  c {2^m} X^{2\widetilde m}$ are in the disk $D(0, e 2^{m/\widetilde m})$.
  % are in the disk $D(0,4e)$, and for $m>2d$, the roots of
  % the polynomial
  % $g(X) + \frac c {2^m} X^{2d}$ are in the disk $D(0, e 2^{m/d})$.
\end{Lemma}
\begin{proof}
  % For a polynomial $a_0 + \cdots + a_n X^n$, the Fujiwara bound ensures
  % that its roots are in the disk $D(0, R)$ where $R = 2 \max_{1 \leq k
  % \leq n} \sqrt[k]{\left|\frac {a_{n-k}}{a_n}\right|}$. In the case
  % $m\leq 2d$, for the polynomial
  % $g(X) + \frac 1 {2^m} X^m$, we have $\sqrt[k]{\left|\frac
  % {c_{m-k}}{c_m}\right|} \leq 2 \left(\frac {m}{m-k}\right)^{(m-k)/k}
  % \leq 2 e^{\frac{m-k}{k}\log\left(1+ \frac k {m-k}\right)}
  % \leq 2e$. When $m \geq 2d$,

  Using the Fujiwara bound on the modulus of the roots of a polynomial
  (Proposition~\ref{pro:fujiwara}) on the polynomial $g(X) + \frac 1 {2^m}
  X^{2\widetilde m}$, we have $\sqrt[k]{\left|\frac
  {c_{2\widetilde m-k}}{c_{2\widetilde m}}\right|}=0$ for $k < \widetilde m$, and for
  $\widetilde m \leq k < 2\widetilde m$ we have
  $\sqrt[k]{\left|\frac {c_{2\widetilde m-k}}{c_{2\widetilde m}}\right|}
  \leq 2^{(m-2\widetilde m+k)/k} \left(\frac {\widetilde m}{2\widetilde
  m-k}\right)^{(2\widetilde m-k)/k}
  \leq 2^{m/\widetilde m-1}e^{\frac{2\widetilde m-k}{k}\log\left(1+
  \frac {k-\widetilde m} {2\widetilde m-k}\right)}
  \leq 2^{m/\widetilde m-1}e^{(k-\widetilde m)/k}\leq e 2^{m/\widetilde
  m-1}$. Finally for $k=2\widetilde m$,
  $\sqrt[2d]{\left|\frac {c_{0}}{c_{2\widetilde m}}\right|} \leq
  2^{m/(2\widetilde m)} \leq
    2^{m/\widetilde m-1}$.
\end{proof}

Then we will use a technical lemma that gives a bound on the derivative of
the difference of an analytic function and a polynomial, given bounds on
their coefficients and the difference of their coefficients.

\begin{Lemma}
  \label{lem:derivative}
  Let $\varphi(x) = \sum_{k=0}^{\infty} \varphi_k x^k$ be
  an analytic series with radius of convergence greater than $2$. Let
  $g$ be a polynomial of degree $m$ and $c$ be a positive real number
  such that: $||\varphi - g||_1 \leq c/2^m$ and $|\varphi_k| < c/2^k$ for all $k>m$.
  Then, for all $x$ in the unit disk we have
  $$|\varphi'(x) - g'(x)| \leq c(m+2)/2^{m}.$$
\end{Lemma}
\begin{proof}
  Using the bounds on the coefficients of $\varphi$ and $g$, we have
  $|g'(x) - \varphi'(x)| \leq \|g'-\varphi'\|_1 \leq m c/2^m
  + \sum_{k=m+1}^\infty (k-m)c/2^k$. The sum $S = \sum_{k=m+1}^\infty k/2^k$ can be bounded
  using the function $\varphi(y) = \sum_{k=m+1}^\infty y^k = y^{m+1}/(1-y)$
  defined for $y$ a real in $[0,1[$. We have $S = 1/2 \varphi'(1/2)$ and
  $\varphi'(y) =
  y^m[(m+1)(1-y)+y]/(1-y)^2$, such that $S = (m+1 - m/2)/2^{m-1} =
  (m+2)/2^m$. Also we have $\sum_{k=m+1}^\infty m/2^k = m/2^m$. This
  leads to $|g'(x) - \varphi'(x)| \leq
  c(m+2)/2^{m}$
\end{proof}

%TODO: blabla
For an analytic function $\varphi$, this allows us to prove that if a polynomial $g$ is a good enough
approximation of $\varphi$, each root of $\varphi$ in the
unit disk is near a root of $g$.

\begin{Lemma}
  \label{lem:approximation}
  Let $\varphi(x) = \sum_{k=0}^{\infty} \varphi_k x^k$ be
  an analytic series with radius of convergence greater than $2$. Let
  $g$ be a polynomial of degree $m$ and $c$ be a positive real number
  such that: $||\varphi - g||_1 \leq c/2^m$ and $|\varphi_k| < c/2^k$ for all $k>m$.
  %and $t$ be a positive integer such that
  %for all $k \geq t$ we have $|f_k|<c/2^k$.
  Let $\zeta$ be a root of $\varphi$ in the unit disk such that $\varphi'(\zeta)
  \neq 0$ and let $\kappa \geq 1/|\varphi'(\zeta)|$.
  % and $s = \max_{z \in D(0,1)}(|f''(z)|)$.

  %Let $g(x) = \sum_{k=0}^{m} g_kx^k$ be a polynomial of degree $m$ such that for all
  %$0 \leq k \leq m-1$ we have $|f_k - g_k| \leq c/2^m$.

  If $2^m/(m+2) \geq 2c\kappa$,
  then $g$ has a root in $D(\zeta,
  2c\kappa m/2^m)$.
  %TODO

  % Let $f(x) = \sum_{k=0}^{\infty} f_k x^k$ be
  % an analytic series with radius of convergence greater than $1$.
  % Assume that there exist $c>0$, $\kappa > 32$, $s > 1$ and an integer $m
  % > 2 \log_2(s\kappa^2)$ such that for all point $z$ in the
  % unit disk:
  % \begin{itemize}
  %   \item $|f(\zeta)|=0$ implies $|f'(z)| > 1/\kappa$,
  %   \item $|f''(z)| < cs$
  %   \item for all $k>m$ we have $|f_k| \leq c/2^k$.
  % \end{itemize}
  % Let $g(x) = \sum_{k=0}^{m-1} g_kx^k$ be a polynomial of degree $m-1$ such that for all
  % $0 \leq k \leq m-1$ we have $|f_k - g_k| \leq c/2^m$.
  % TODO: change to property for one root

  % Then, for each root $\zeta$ of $f$ in the unit disk, $f$ has no other
  % root in $D(\zeta, 1/(2s\kappa))$ and $g$ has a root in the disk
  % $D(\zeta, 1/(16s\kappa))$. Moreover, if $g$ has a root $\eta$ in the
  % unit disk, then $f$ has a root in the disk $D(\eta, 1/(16s\kappa))$.
\end{Lemma}
\begin{Remark}
  If $m \geq 10$, the inequality $2^m/(m+2) \geq 2c\kappa$ holds as soon
  as $m \geq 2 \log_2(c\kappa)$.
\end{Remark}
\begin{proof}
  % First, let $\eta$ be a root of $g$ in the unit disk.  Then $|f(\eta)|
  % = |f(\eta)-g(\eta)| \leq c\frac {m+1} {2^m} + \frac{c} {2^m} \leq
  % c\frac{m+2}{2^m}$ using the bounds on the difference of the
  % coefficients of $f$ and $g$. 
  % In particular, with the lower bound on $m$, we have $m \geq
  % \log_2(\kappa) + m/2$ and $m \geq 20$ since $s\kappa^2 \geq
  % 2^{10}$. This implies that $|f(\eta) \leq c(m+2)/2^m \leq c/(s\kappa^2) \cdot
  % (m+2)/2^{(m/2)} \leq c/(32s\kappa^2)$.  This implies that $|f'(\eta)|
  % \geq c/\kappa$. In turn, we have $\beta = |f(\eta)/f'(\eta)| \leq
  % 1/(32s\kappa)$, and $K = \max_{z \in U}(|f''(z)/f'(\eta)|) \leq
  % s\kappa$. Thus, $2\beta K \leq 1/16 \leq 1$. Using Kantorovich's theory
  % \cite[Theorem~88]{Dbook06}, this
  % ensures that $f$ has a root in $D(\eta, 2\beta)$ which implies that
  % $f$ has a root in the disk $D(\eta, 1/(16s\kappa))$. Moreover, using
  % Kantorovich's theory again \cite[Theorem~88]{Dbook06}, since $2K \leq
  % 2s\kappa$, this implies that $\zeta$ is the only
  % root of $f$ in the disk $D(\zeta, 1/(2s\kappa))$.

  Using Proposition~\ref{pro:crude}, if $g'(\zeta) \neq 0$, then $g$ has a root in the disk
  $D(\zeta, m g(\zeta)/g'(\zeta))$. Since $\zeta$ is in the unit disk,
  $|g(\zeta)| = |g(\zeta)-\varphi(\zeta)| \leq \|g-\varphi\|_1 \leq c/2^m$. For the
  derivative, we have $|g'(\zeta)| \geq |\varphi'(\zeta)| - |g'(\zeta) -
  \varphi'(\zeta)| \geq 1/\kappa - |g'(\zeta)-\varphi'(\zeta)|$. The difference
  between the derivative of $g$ and $\varphi$ can be bounded using
  Lemma~\ref{lem:derivative} by $|g'(\zeta) - \varphi'(\zeta)| \leq
  c(m+2)/2^{m}$.  Since $2^m/(m+2) \geq 2c\kappa$, this implies
  $|g'(\zeta)| \geq 1/(2\kappa)$, which allows us to conclude.

\end{proof}

Finally, to prove that Algorithm~\ref{alg:rootfinding} terminates, we
will need the following lemma that guarantees that the criterion
of Lemma~\ref{lem:distinct} will be satisfied for a small enough
approximation.

\begin{Lemma}
  \label{lem:termination}
  Let $\varphi(x) = \sum_{k=0}^{\infty} \varphi_k x^k$ be
  an analytic series with radius of convergence greater than $2$ and
  $\zeta$ be a root of $\varphi$ in the unit disk such that $\varphi'(\zeta) \neq
  0$. Let $\kappa = 1/|\varphi'(\zeta)|$ and $s$ be positive real
  %greater than $|f'(\zeta)|$ and
  greater than $|\varphi''(y)|$ for all $y$ in the disk $D(0,1)$. Then, for
  any positive real $\varepsilon \leq 1/[23(s\kappa^2+\kappa)]$ and all $x\in
  D(\zeta, \kappa \varepsilon)$:
  $$ q := 10\frac{s(|\varphi(x)|+\varepsilon)}{(|\varphi'(x)|-\varepsilon)^2} < 1.$$
  % $g$ be a polynomial of degree $m > 12$ and $c$ be a positive real number
  % such that: $||f - g||_1 \leq c/2^m$ and $|f_k| < c/2^k$ for all $k>m$.
  % %and $t$ be a positive integer such that
  % %for all $k \geq t$ we have $|f_k|<c/2^k$.
  % Let $\zeta$ be a root of $f$ in the unit disk such that $f'(\zeta)
  % \neq 0$ and let $\kappa = 1/|f'(\zeta)|$.

  %TODO: update disk radius to $\kappa \varepsilon$ or handle it in next lemma.
\end{Lemma}
\begin{proof}
  Using Taylor expansion at $\zeta$, we have $|x-\zeta| \leq
  \kappa\varepsilon$, and thus $s(|\varphi(x)| + \varepsilon) \leq
  s(\kappa \varepsilon \varphi'(\zeta) + \frac 1 2 \kappa^2\varepsilon^2s + \varepsilon)$.
  Similarly, $|\varphi'(x)| - \varepsilon > |\varphi'(\zeta)| - \kappa\varepsilon s -
  \varepsilon$. Factoring out $\varepsilon$ in the numerator, and
  $|\varphi'(\zeta)|$ in the
  denominator, this leads to $q \leq 10 s\kappa^2\varepsilon
  \frac{2+\frac 1 2
\varepsilon\kappa^2s}{\left(1-\varepsilon(\kappa^2s+\kappa)\right)^2}$.
  Since $\varepsilon \leq 1/[23(s\kappa^2+\kappa)]$, the numerator is less
  than $93/46$ and the denominator is greater than $22^2/23^2$, such that $q
  \leq \frac{10}{23} \frac{93 \cdot 23^2}{46 \cdot 22^2} \leq 1$.%  25 s \kappa \varepsilon \leq 1$ using $\varepsilon \leq 1/(25s\kappa)$.

  % Letting $b =
  % \kappa^2
  % s+1$, this implies $q \leq \frac{10}{\kappa} b\varepsilon \frac{1+\frac 1 2
  % b\varepsilon}{1-b\varepsilon}$.

  % Letting $K = (s+1)\kappa$, this implies that $q \leq
  % 10 K \varepsilon \frac{1 + \varepsilon K}{1 - \varepsilon K}$.
  % Finally, if $\varepsilon \leq 1/[12 (s+1)\kappa] = 1 /(12K)$, then $q \leq \frac {10}{12}
  % \frac{13}{11} = \frac{130}{132} < 1$.
\end{proof}

% TODO: maybe note on link between separation bound and f''/f'

\subsection{Proof of Theorem~\ref{thm:rootfinding}}
\label{sec:proofrootfinding}

We can now prove the main theorem bounding the bit complexity of
Algorithm~\ref{alg:rootfinding}.
%TODO: add lemma on bound on $\|g_k\|_1 \leq \|f\|_1\max(1,
%(\gamma_k+\rho_k)^d) \leq \|f\|_1(1+m/(6ed))^d \leq \|f\|_1e^{m/(6e)} \leq \|f\|_12^{m/11}$ in the detailed section on data structure.
%\begin{proof}[Proof of Theorem~\ref{thm:rootfinding}]
  We split our proof in three part. First the correctness, then the
  termination and finally a bound on the complexity of
  Algorithm~\ref{alg:rootfinding}.
  
  %\emph{Correctness.}
  \paragraph{Correctness.}
  First, the correctness of Algorithm~\ref{alg:rootfinding} follows from
  Lemma~\ref{lem:distinct}. Indeed, using Lemma~\ref{lem:second} and
  Lemma~\ref{lem:derivative}, each disk added to $L$ satisfies the
  condition of Lemma~\ref{lem:distinct} and contains a unique root of
  $f$. Then, if the algorithm terminates, it returns a list of $d$
  pairwise distinct disks, containing a root of $f$ each, such that the
  result is correct.
  %TODO: expand first paragraph, in particular show that we have the
  %correct bound on $f''$.

  \paragraph{Termination.}
  For the termination of Algorithm~\ref{alg:rootfinding}, we fix $m$ and we will use
  Lemma~\ref{lem:approximation} and~\ref{lem:termination} to show that
  for $m$ sufficiently large, Algorithm~\ref{alg:rootfinding} terminates.
  First, using Lemma~\ref{lem:approximation} to bound the distance between a root of $\varphi := f(a(X))$ and the closest root of
  $\widetilde h$, we need a bound $\|\varphi - \widetilde h\|_1$
  and a bound on the condition number of
  $\varphi$. The first bound comes from
  %$\|\varphi_k -\widetilde
  %h_k\|_1 \leq \|\varphi_k - g_k\|_1 + \|g_k - h_k\|_1 + \|h_k -
  %\widetilde h_k\|_1 \leq 3\|f\|_1/2^m + \|f\|_1/2^m +
  %\|h_k\|_1/2^{1.1m}$.
  $\|\varphi -\widetilde h\|_1 \leq \|\varphi - h\|_1 + \|h -
  \widetilde h\|_1 \leq \|\varphi - h\|_1 + \|h\|_1/2^{1.1m}
  \leq \|\varphi - h\|_1(1+1/2^{1.1m}) + \|\varphi\|_1/2^{1.1m}
  $.
  Using Lemma~\ref{lem:second}, we have $\|\varphi\|_1 \leq \|f\|_1
  2^{m/11}$, such that $ \|\varphi -\widetilde h\|_1 \leq
  \|\varphi - h\|_1(1+1/2^{1.1m}) + \|f\|_1/2^{m}$. Moreover,
  $\|\varphi-h\|_1 \leq \|\varphi - g\|_1 + \|g - h\|_1 \leq
  3\|f\|_1/2^m + \|f\|_1/2^m$ and $1+1/2^{1.1m} \leq 5/4$ for $m\geq 2$.
  This leads to $\|\varphi -\widetilde h\|_1 \leq 6\|f\|_1/2^m$.
  For the bound on the condition number, since
  $a$ is of the form $a(X) = (\gamma + \rho X) e^{i\alpha}$,
  this implies that for all $x$ in the unit disk $|\varphi'(x)| = \rho
  |f'(x)| \geq
  \min(1,\frac m {2ed}) |f'(x)|$, such that $\kappa_1(\varphi) \leq
  \max(1, \frac{2ed}{m}) \kappa_1(f)$.
  Letting $\kappa = 2ed\kappa_1(f)$, we have that for each root $\zeta_j$ of
  $\varphi$, if $m > 2 \log_2(6\|f\|_1\kappa)$, Lemma~\ref{lem:approximation}
  implies that there exists a root $z_j$ of $\widetilde h$ in the disk
  $D(\zeta_j, \mu\kappa)$, where $\mu = 12\|f\|_1m/2^m$.

  %TODO: change the split of paragraph, explain a bit better the main
  %directions of the proof

  We will now use this property with Lemma~\ref{lem:termination} to show
  that the criterion computed on line~\ref{step:criterion} of
  Algorithm~\ref{alg:rootfinding} will eventually be satisfied.
  Using the notations of Algorithm~\ref{alg:rootfinding}, we show
  that the criterion $10\beta K \leq 1$ will be satisfied for all roots
  of $f$ for $m$ large enough.  Let $s = \|f\|_1d^22^{m/10}$.
  Using the bound on $|f-g|$ given by
  Definition~\ref{def:approximation} and the bound on 
  $|f'-g'|$ given by Lemma~\ref{lem:derivative}, we have $10\beta K \leq
  10\frac{s(|f(z_j)|+2\varepsilon)}{(|f'(z_j)|-2\varepsilon)^2}$.
  Moreover $z_j$ is in the disk $D(\zeta_j, \mu\kappa)$, with $\mu \geq
  6\|f\|_1(m+2)/2^m = 2\varepsilon$. From
  Lemma~\ref{lem:termination}, we can conclude that $10\beta K$ is
  smaller than $1$ for $\mu \leq 1/[23(s\kappa^2+\kappa)]$,
  that is for $12\|f\|_1m/2^m \leq
  1/[23(\|f\|_1m^22^{m/10}\kappa^2+\kappa)]$, which holds as soon as
  $12\|f\|_1m^3/2^{\frac 9{10}m} \leq 1/[23(\|f\|_1\kappa^2+\kappa)]$.
  Note that for $m\geq40$, $m/2^{\frac9{10} m}$ is smaller than
  $1/2^{m/2}$, such that $10\beta\kappa \leq 1$ for all $m >
  2\log_2(276(\|f\|_1^2d^2\kappa^2+\kappa))$ and the
  algorithm terminates after $O(\log(\|f\|_1\kappa_1(f)))$ iterations
  of the main loop.

  \paragraph{Complexity bound.}
  %TODO: reread and clarify the proof.

  First, at each iteration of the main \verb/while/ loop, computing the
  $m$-hyperbolic approximation costs $\widetilde O(dm)$ bit operations.
  Then for a fixed $m \leq 2d$, the approximate factorization 
  is called $O(d/m)$ times on polynomials of degree $m$, with $\widetilde O(m^2)$ bit operations for
  each call, using Proposition~\ref{pro:pan}. With Remark~\ref{rem:pan}, this bound holds for
  polynomial that have all their roots of modulus less than $e2^{m/d}$.
  %If $g = c_0 + \cdots + c_m X^m$, the Fujiwara bound (\cite{} Fujiwara) on the absolute value
  %of the roots is $2 \max_{1 \leq k \leq m} \sqrt[k]{\frac {c_{m-k}}{c_m}}$.
  By Lemma~\ref{lem:coefficients}, the coefficients of the polynomial
  $h$ satisfy the condition of Lemma~\ref{lem:compact} and we conclude
  that all its
  roots of $h$ are included in $D(0, e2^{m/d})$.
  % When $m > 2d$, the
  % approximate factorization is called a constant number of times and
  % Lemma~\ref{lem:compact} guarantees us that all the root of $h_k$ are
  % in the disk $D(0, e2^{m/d})$. With Remark~\ref{rem:pan}, the
  Thus the approximate factorization can be computed within $\widetilde O(dm)$.
  Thus, for all cases, the total cost for the approximate factorization
  in an iteration of the \verb/while/ loop is in $\widetilde O(dm)$ bit
  operations.  After that, for each $g$, we need to evaluate $K$ and
  $\beta$ up to a precision in $O(\log(\|f\|_1d) + m)$. This can be done
  using state-of-the-art fast approximate multipoint evaluation in
  $\widetilde O(m(m+\log(\|f\|_1d))$ for all the approximate roots of
  $\widetilde h$ using Proposition~\ref{pro:evaluation}. This amounts to a total of
  $\widetilde O(d(m+\log(\|f\|_1d))$ bit operations for the steps
  \ref{num:startfor} to \ref{num:endfor}. If $m \geq d$, then the
  factorization is called a constant number of times on polynomials of
  degree $d$, for a total cost in $\widetilde O(dm)$ bit operations, and
  all the multipoint evaluations will cost a total of $\widetilde
  O(d(m+\log(\|f\|_1\kappa))$ bit operations. Finally, removing
  duplicate solutions can be done in $\widetilde O(d)$ operations.
  Indeed, by construction, given a box of $B$ in a disk $D$ of the
  $N$-hyperbolic covering, the number of times that it
  appears in $L$ is bounded by the maximal number of disks of the $N$-hyperbolic
  covering that intersects $D$, that is $10$ (see
  Section~\ref{sec:range}). Thus, using
  Proposition~\ref{pro:rectangles}, this ensures that each query to
  detect a duplicate will cost at most $\widetilde O(\log d)$ operations,
  and removing all the duplicates will cost at most $\widetilde O(d)$
  bit operations.  In total, the costs is $\widetilde O(dm)$ per iteration of the
  \verb/while/ loop. Since $m$ is doubled at each iteration, the cost is
  the same as the cost of the last iteration, that is in $\widetilde
  O(d\log(\|f\|_1\kappa_1(f)))$ bit operations.

%Proof of the lemma for the roots of the approximated polynomial

%Proof of time complexity

\section{Lower bound on the condition number}
\label{sec:condition}

The lower bound on the condition number is a consequence of
Lemma~\ref{lem:approximation}, that gives a bound on the distance
between the roots of two polynomials with close enough coefficients,
applied on the polynomials from the adapted hyperbolic approximation.
Essentially, the idea is that for any $m$ greater or
equal to
function of $\kappa_1(f)$ and for any pair $(g, a)$ of an
$m$-hyperbolic approximation of $f$, the number of roots of $g$ is
greater than the number of roots of $f(a(X))$ in the disk
$a(D(0,1))$. In particular, this property allows us to deduce a lower bound on
$\kappa_1(f)$ depending on the number of roots of $f(a(X))$ in the
disk $a(D(0,1))$.

\begin{proof}[Proof of Theorem~\ref{thm:condition}]
  %Let $s = d^2\|f\|_1$, an upper bound on $f''(x)$ for all $x$ in the unit
  %disk, and let $c = $. Let $M$ be the smallest integer such that $2^M/M \geq 4sc\kappa^2$.
  %TODO.

  Let $D$ be a disk where the number $m$ of solutions of $f$ is maximal.
  Without restriction of generality, using
  Remark~\ref{rem:coefficients}, we can assume that the absolute value of the center of $D$ is less
  than $1$. Up to a change of variable $f(uX)$, where $u$ is a complex
  number of modulus $1$, we can assume that the center of $D$ is a
  positive real number.
  %Let $m$ be a positive integer such that $m \leq M$. In this case, letting
  Letting $N = \lceil\log_2(3ed/m)\rceil$, we can see easily that $D$ is
  included in a disk $D(\gamma, \rho)$ of the $N$-hyperbolic covering of the unit disk.
  Let $\varphi(X) = f(\gamma + \rho X)$ and let $g(X)$ be the polynomial
  of degree $m-1$ obtained by truncating $\varphi$ at order $m-1$. By
  Lemma~\ref{lem:coefficients}, the coefficients of $\varphi$ for
  degree $\ell\geq m$ are less than $\|f\|_1/2^\ell$. Using the
  bounds in Lemma~\ref{lem:second}, we have for all $x$ in the unit disk $\|\varphi''(x)\|_1
  \leq \|f\|_1m^22^{m/11}$. Moreover $\varphi' = \rho f'$, and $\rho
  \geq \frac m {4ed}$, such that
  $\kappa_1(\varphi) \leq \frac{4ed}{m}\kappa_1(f)$. Let $c =
  \|f\|_1$, $\kappa = \frac{4ed} m \kappa_1(f)$ and $s =
  \|f\|_1m^22^{m/11}$. We can now prove by contradiction that the number
  of roots of $g$ exceeds its degree if $2^m \geq \max(2c\kappa(m+2),
  4cs\kappa^2m)$. Indeed in this case, by
  Lemma~\ref{lem:approximation}, for each root $\zeta$ of $\varphi$, the
  polynomial $g$ has a root in $D(\zeta, 2c\kappa m/2^m)$. Moreover,
  using Remark~\ref{rem:distinct}, for
  all the roots $\zeta$ of $\varphi$ in the unit disk, the disks
  $D(\zeta, 2c\kappa m /2^m) \subset D(\zeta, 1/(2s\kappa))$ are pairwise distinct, such that $g$ has
  at least $m > m-1$ roots. Thus, to avoid this contradiction, we thus
  have $\max(2c\kappa(m+2), 4cs\kappa^2m) > 2^m$. That is, we must have
  either $2 \|f\|_1\kappa_1(f) (m+2) 4ed/m > 2^m$ or
  $\kappa_1(f)^2\|f\|_1^2(4ed)^22^{m/11}m > 2^m$. Equivalently, this amounts
  to $\|f\|_1\kappa_1(f) > \min\left(\frac 1 {8ed} \frac m{m+2} 2^{m}, \frac
  1 {4ed\sqrt{m}} 2^{5m/11}\right)$. For $m \geq 3$, this implies
  $\|f\|_1\kappa_1(f) \geq \frac 1 {4ed\sqrt{m}} 2^{5m/11}$. Finally,
  for each root $\zeta$ of $\varphi$ with $|\zeta| < 1$ we have
  $\|f\|_1\kappa_1(f)/|\zeta| > \|f\|_1\kappa_1(f)$, which allows us to
  conclude.
  %Letting $x = \|f\|_1\kappa_1(f)$, this becomes: $x^2(4ed)^22^{m/11} + 8x ed/m - 2^m/(m+2) >0$.
\end{proof}

\section*{Acknowledgement}
The author thanks the anonymous reviewers for their thoughtful comments
on this work and on a previous related work.

\appendix

\section{Source code}
\label{apx:source}

For the reproducibility of our experiments, and to demonstrate the
conciseness of our implementation, we report here the full source code of our root
solver
\verb/HCRoots/, along with an implementation of the multipoint evaluation algorithm
(\verb/hceval.py/), both available on a public gitlab server \cite{Msoftware21}.

%TODO: mention availability link git or zenodo.

%along with the code we used to measure the
%timings of Figure~\ref{fig:benchmark}. Note that we took care to
%configure \verb/MPSolve/ to use the same precision as \verb/HCRoots/.

\definecolor{mygreen}{rgb}{0,0.6,0}
\definecolor{mygray}{rgb}{0.5,0.5,0.5}
\definecolor{mymauve}{rgb}{0.58,0,0.82}
\lstset{ 
  backgroundcolor=\color{white},   % choose the background color; you must add \usepackage{color} or \usepackage{xcolor}; should come as last argument
  basicstyle=\ttfamily\footnotesize,        % the size of the fonts that are used for the code
  breakatwhitespace=false,         % sets if automatic breaks should only happen at whitespace
  breaklines=true,                 % sets automatic line breaking
  captionpos=b,                    % sets the caption-position to bottom
  commentstyle=\color{mygreen},    % comment style
  deletekeywords={...},            % if you want to delete keywords from the given language
  escapeinside={\%*}{*)},          % if you want to add LaTeX within your code
  extendedchars=true,              % lets you use non-ASCII characters; for 8-bits encodings only, does not work with UTF-8
  firstnumber=1,                   % start line enumeration with line 1000
  frame=single,	                   % adds a frame around the code
  keepspaces=true,                 % keeps spaces in text, useful for keeping indentation of code (possibly needs columns=flexible)
  keywordstyle=\color{blue},       % keyword style
  language=Octave,                 % the language of the code
  morekeywords={*,...},            % if you want to add more keywords to the set
  numbers=left,                    % where to put the line-numbers; possible values are (none, left, right)
  numbersep=5pt,                   % how far the line-numbers are from the code
  numberstyle=\tiny\color{mygray}, % the style that is used for the line-numbers
  rulecolor=\color{black},         % if not set, the frame-color may be changed on line-breaks within not-black text (e.g. comments (green here))
  showspaces=false,                % show spaces everywhere adding particular underscores; it overrides 'showstringspaces'
  showstringspaces=false,          % underline spaces within strings only
  showtabs=false,                  % show tabs within strings adding particular underscores
  stepnumber=2,                    % the step between two line-numbers. If it's 1, each line will be numbered
  stringstyle=\color{mymauve},     % string literal style
  tabsize=2,                       % sets default tabsize to 2 spaces
  title=\lstname                   % show the filename of files included with \lstinputlisting; also try caption instead of title
}

\begin{lstlisting}[language=Python]
#    Copyright (C) 2021 Guillaume Moroz <guillaume.moroz@inria.fr>
# This program is free software: you can redistribute it and/or modify
# it under the terms of the GNU General Public License as published by
# the Free Software Foundation, either version 2 of the License, or
# (at your option) any later version.
import numpy as np

# Compute the disks of a hyperbolic covering
def disks(d, m):
    N = np.math.ceil(np.log2(3*np.e*d/min(m-1,d)))
    r = 1 - 1/2**(np.arange(N+1))
    r[-1] = 1
    gamma = 1/2*(r[1:] + r[:-1])
    rho = 3/4*(r[1:] - r[:-1])
    K = np.ceil(3*np.pi*r[1:]/(np.sqrt(5)*rho)).astype(int)
    K[0] = 4
    return gamma, rho, K

# Compute the m-hyperbolic approximation
def hyperbolic_approximation(coeffs, m=30):
    d = coeffs.shape[-1]
    shape = coeffs.shape[:-1]
    gamma, rho, K = disks(d, m)
    N = gamma.size
    Kmax = ((d-1)//K.max()+1)*K.max()
    r = rho/gamma
    D = np.arange(d)
    G = np.zeros(shape + (N, Kmax, m), dtype='complex128')
    P = gamma[:, np.newaxis]**D * coeffs[..., np.newaxis, :]
    G[...,0] = np.fft.fft(P, Kmax)
    for i in range(m-1):
        P *= (D-i)/(i+1) * r[:, np.newaxis]
        G[..., i+1] = np.fft.fft(P[...,i+1:], Kmax)
    return G, gamma, rho, K

# Solve polynomials of small degree
def solve_small(p, m=30, guarantee=True, e=0):
    result = [np.empty(0)]*p.shape[0]
    abs_p = np.abs(p)
    nosol = abs_p[:,0] > abs_p[:,1:].sum(axis=-1)
    unksol = ~nosol
    sols = list(map(np.polynomial.polynomial.polyroots, p[unksol]))
    for i,j in enumerate(np.flatnonzero(unksol)):
        result[j] = sols[i][np.abs(sols[i])<=1]
    if guarantee:
        validate(result, p, e)
    return result

# Guarantee that there is a unique solution nearby
def validate(sols, p, e):
    nonempty = [i for i,x in enumerate(sols) if x.size>0]
    p0 = p[nonempty]
    p1 = np.polynomial.polynomial.polyder(p0, axis=-1)
    p2 = np.polynomial.polynomial.polyder(p1, axis=-1)
    s = np.linalg.norm(p2, 1, axis=-1)
    for i, j in enumerate(nonempty):
       q = 10*s[i]*(np.abs(np.polynomial.polynomial.polyval(sols[j], p0[i]))+e)/\
                   (np.abs(np.polynomial.polynomial.polyval(sols[j], p1[i]))-e)**2
       sols[j] = sols[j][q <= 1]


# Solve using truncated polynomials
def solve_piecewise(G, gamma, rho, K, m=30, rtol=8, guarantee=True, e=0):
    result = np.array([],dtype='complex128')
    Kmax = G.shape[1]
    for p, g, r, Kn in zip(G,gamma,rho,K):
        step = (Kmax-1)//(Kn-1)  # step * (Kn-1) < Kmax
        w = np.exp(-2j*np.pi*np.arange(0,Kmax,step)/Kmax)
        sols = solve_small(p[::step], m, guarantee, e)
        for i in range((Kmax-1)//step + 1):
            sols[i] = g*w[i] + r*sols[i]
            result = np.append(result, sols[i])
    rounded = np.round(result, decimals=-int(np.log10(rtol)))
    _, ind = np.unique(rounded, return_index=True)
    return result[ind]

# truncate and solve a polynomial over the complex
def solve(p, m=30, rtol=None, guarantee=True):
    rtol =  max(3*2**(-m), 2**-35) if rtol is None else rtol
    dtype = p.dtype if hasattr(p, 'dtype') else 'complex128'
    p = np.trim_zeros(p, 'b')
    coeffs = np.zeros((2, len(p)), dtype=dtype)
    coeffs[0] = p
    coeffs[1] = coeffs[0,::-1]
    G, gamma, rho, K = hyperbolic_approximation(coeffs, m)
    e = 3*np.linalg.norm(coeffs[0], 1)*(m+2)/2**m
    sols = solve_piecewise(G[0], gamma, rho, K, m, rtol, guarantee, e)
    invsols = solve_piecewise(G[1], gamma, rho, K, m, rtol, guarantee, e)
    result = np.concatenate([sols , 1/invsols])
    rounded = np.round(result, decimals=-int(np.log10(rtol)))
    _, ind = np.unique(rounded, return_index=True)
    return result[ind]

\end{lstlisting}
\begin{lstlisting}[language=Python]
#    Copyright (C) 2021 Guillaume Moroz <guillaume.moroz@inria.fr>
# This program is free software: you can redistribute it and/or modify
# it under the terms of the GNU General Public License as published by
# the Free Software Foundation, either version 2 of the License, or
# (at your option) any later version.
import numpy as np

# Compute the disks of a hyperbolic covering
def disks(d, m):
    N = np.math.ceil(np.log2(3*np.e*d/min(m-1,d)))
    r = 1 - 1/2**(np.arange(N+1))
    r[-1] = 1
    gamma = 1/2*(r[1:] + r[:-1])
    rho = 3/4*(r[1:] - r[:-1])
    K = np.ceil(3*np.pi*r[1:]/(np.sqrt(5)*rho)).astype(int)
    K[0] = 4
    return gamma, rho, K

# Compute the m-hyperbolic approximation
def hyperbolic_approximation(coeffs, m=30):
    d = coeffs.shape[-1]
    shape = coeffs.shape[:-1]
    gamma, rho, K = disks(d, m)
    N = gamma.size
    Kmax = ((d-1)//K.max()+1)*K.max()
    r = rho/gamma
    D = np.arange(d)
    G = np.zeros(shape + (N, Kmax, m), dtype='complex128')
    P = gamma[:, np.newaxis]**D * coeffs[..., np.newaxis, :]
    G[...,0] = np.fft.fft(P, Kmax)
    for i in range(m-1):
        P *= (D-i)/(i+1) * r[:, np.newaxis]
        G[..., i+1] = np.fft.fft(P[...,i+1:], Kmax)
    return G, gamma, rho, K

# Get the indices to match points to the corresponding disk
def get_indices(N, Kmax, points):
    module_indices = np.zeros(points.shape, int)
    angle_indices = np.zeros(points.shape, int)
    apoints = np.abs(points)
    big = apoints > 1-1/2**(N-1)
    small = apoints < 1/2
    middle = ~small & ~big
    module_indices[middle] = np.log2(1/(1-apoints[middle])).astype(int)
    module_indices[big] = N-1
    angle_indices[:] = (0.5 - np.angle(points)*(Kmax/(2*np.pi)) % Kmax).astype(int)
    return module_indices, angle_indices

# Evaluate the points in a unit disk
def eval_unitdisk(G, gamma, rho, points):
    N, Kmax, m = G.shape
    m_ind, a_ind = get_indices(N, Kmax, points)
    shift_points = (points - gamma[m_ind]*np.exp(-2j*np.pi*a_ind/Kmax))/rho[m_ind]
    res = np.polynomial.polynomial.polyval( shift_points, G[m_ind, a_ind].T, tensor=False)
    return res

# Evaluate the points in the complex plane
def eval_hyperbolic_approximation(covering, points):
    G, gamma, rho, d = covering
    points = np.array(points)
    apoints = np.abs(points)
    inpoints = points[apoints <= 1]
    outpoints = points[apoints > 1]
    res = np.zeros(points.size, dtype='complex128')
    res[apoints <= 1] = eval_unitdisk(G[0], gamma, rho, inpoints)
    res[apoints >  1] = eval_unitdisk(G[1], gamma, rho, 1/outpoints)*outpoints**d
    return res

# Compute the hyperbolic approximation used for multi-point evaluation
def get_hyperbolic_approximation(p, m=30):
    d = len(p)-1
    dtype = p.dtype if hasattr(p, 'dtype') else 'complex128'
    coeffs = np.zeros((2, d+1), dtype=dtype)
    coeffs[0] = p
    coeffs[1] = coeffs[0,::-1]
    G, gamma, rho, K = hyperbolic_approximation(coeffs, m)
    covering = G, gamma, rho, d
    return covering

# Compute the hyperbolic approximation and evaluate the points
def eval(p, points, m=30):
    covering = get_hyperbolic_approximation(p, m)
    res = eval_hyperbolic_approximation(covering, points)
    return res


\end{lstlisting}

%\newpage

%\lstinputlisting[language=Python]{benchmark.py}

%\bibliographystyle{ACM-Reference-Format}
\bibliographystyle{plain}
\bibliography{bibliography}

\end{document}